\documentclass[aps,prd,showpacs,amssymb,superscriptaddress,nofootinbib,onecolumn,notitlepage]{revtex4-1}
\usepackage{amsmath}
\usepackage{amsfonts,amsthm,amssymb,setspace,amsmath,appendix,verbatim}
\usepackage{bbm,bm,amsbsy,dsfont}
\usepackage{graphicx,epsfig}
\usepackage{color}
\usepackage[colorlinks]{hyperref}
\usepackage[figure,table]{hypcap}
\usepackage{enumerate}
\usepackage[resetlabels]{multibib}
\usepackage[usenames,dvipsnames]{xcolor}
\usepackage{multirow}
\usepackage{hhline}
\usepackage{times}
\newtheorem{theorem}{Theorem}
\newtheorem{lemma}{Lemma}

\newtheorem{definition}{Definition}

\newtheorem{corollary}{Corollary}

\newcommand{\specialcell}[2][c]{%
  \begin{tabular}[#1]{@{}c@{}}#2\end{tabular}}

\newcommand{\be}{\begin{equation}}
\newcommand{\ee}{\end{equation}}
\newcommand{\ba}{\begin{align}}
\newcommand{\ea}{\end{align}}
\newcommand{\tr}{\textup{tr}}
\newcommand{\la}{\langle}
\newcommand{\ra}{\rangle}
\newcommand{\bo}{{\Theta}}

\newcommand{\id}{{\mathbbm{1}}}

\newcommand{\ketbra}[2]{\ensuremath{| #1\rangle\!\!\langle #2 |} }

\newcommand{\ot}{\otimes}
\newcommand{\diff}[2]{\frac{d #1}{d #2}}
\newcommand{\seconddiff}[2]{\frac{d^2 #1}{d {#2}^2}}
\newcommand{\onecold}{\textup{c}}

\newcommand{\hot}{\textup{Hot}}
\newcommand{\cold}{\textup{Cold}}
\newcommand{\batt}{\textup{W}}
\newcommand{\mach}{\textup{M}}
\newcommand{\CW}{\textup{ColdW}}
\newcommand{\CMW}{\textup{ColdMW}}
\newcommand{\total}{\textup{ColdHotMW}}

\newcommand{\Wext}{W_\textup{ext}}
\newcommand{\Walpha}{W_\alpha}
\newcommand{\alpst}{\alpha^*}

\usepackage[margin=0.7in]{geometry}
\begin{document}
	\title{Surpassing the Carnot Efficiency by extracting imperfect work}

	\author{Nelly Huei Ying Ng}
	\affiliation{QuTech, Delft University of Technology, Lorentzweg 1, 2611 CJ Delft, Netherlands}
	\affiliation{Centre for Quantum Technologies, National University of Singapore, 117543 Singapore}
	
	\author{Mischa Woods}
	\affiliation{University College of London, Department of Physics \& Astronomy, 
				London WC1E 6BT, United Kingdom}
	\affiliation{QuTech, Delft University of Technology, Lorentzweg 1, 2611 CJ Delft, Netherlands}

	\author{Stephanie Wehner}
	\affiliation{QuTech, Delft University of Technology, Lorentzweg 1, 2611 CJ Delft, Netherlands}
	\affiliation{Centre for Quantum Technologies, National University of Singapore, 117543 Singapore}
	\date{\today}

	\begin{abstract}
	A suitable way of quantifying work for microscopic quantum systems has been constantly debated in the field of quantum thermodynamics. 
	One natural approach is to measure the average increase in energy of an ancillary system, called the battery, after a work extraction protocol. 
	The quality of energy extracted is usually argued to be good by quantifying higher moments of the energy distribution, or by restricting the amount of entropy to be low. 
	This limits the amount of heat contribution to the energy extracted, but does not completely prevent it. 
	We show that the definition of ``work" is crucial. If one allows for a definition of work that tolerates a non-negligible entropy increase in the battery, then a small scale heat engine can possibly exceed the Carnot efficiency. This can be achieved even when one of the heat baths is finite in size.
	\end{abstract}
	\maketitle
	
	\section*{Introduction}
	Given physical systems where energy is only present in its most disordered form (heat), how efficiently can one convert such heat and store it as useful energy (work)? 
	This question lies at the foundation of constructing heat engines, like the steam engine. 
	Though nearly two centuries old, it remains one of central interest in physics, and can be applied in studying a large variety of systems, from naturally arising biological systems to intricately engineered ones. 
	Classically it is known that a heat engine cannot perform at efficiencies higher than a theoretical maximum known as the Carnot efficiency, which is given by 
	\begin{equation}
	\eta_C = 1-\frac{T_\cold}{T_\hot},
	\end{equation}
	$T_\cold, T_\hot$ being the temperatures of the heat reservoirs at which the engine operates between. This fundamental limit on efficiency can be derived as a consequence of the second law of thermodynamics, which is regarded as one of the ``most perfect laws in physics" \cite{LiebYngvason}.
	
	Recent advancements in the engineering and control of quantum systems have pushed the applicability of conventional thermodynamics to its limits.
	In particular, instead of large scale machines that initially motivated the study of thermodynamics, we are now able to build nanoscale quantum machines. A quantum heat engine (QHE) is a machine that performs the task of work extraction when the involved systems are not only extremely small in size/particle numbers, but also subjected to the laws of quantum physics. 
	Such studies are highly motivated by the prospects of designing small, energy efficient machines applicable to state-of-the-art devices, particularly those relevant for quantum computing and information processing.
	The question then arises: how efficient can these machines be?
	\begin{center}
		\begin{figure}[h!]
		\includegraphics[scale=0.35]{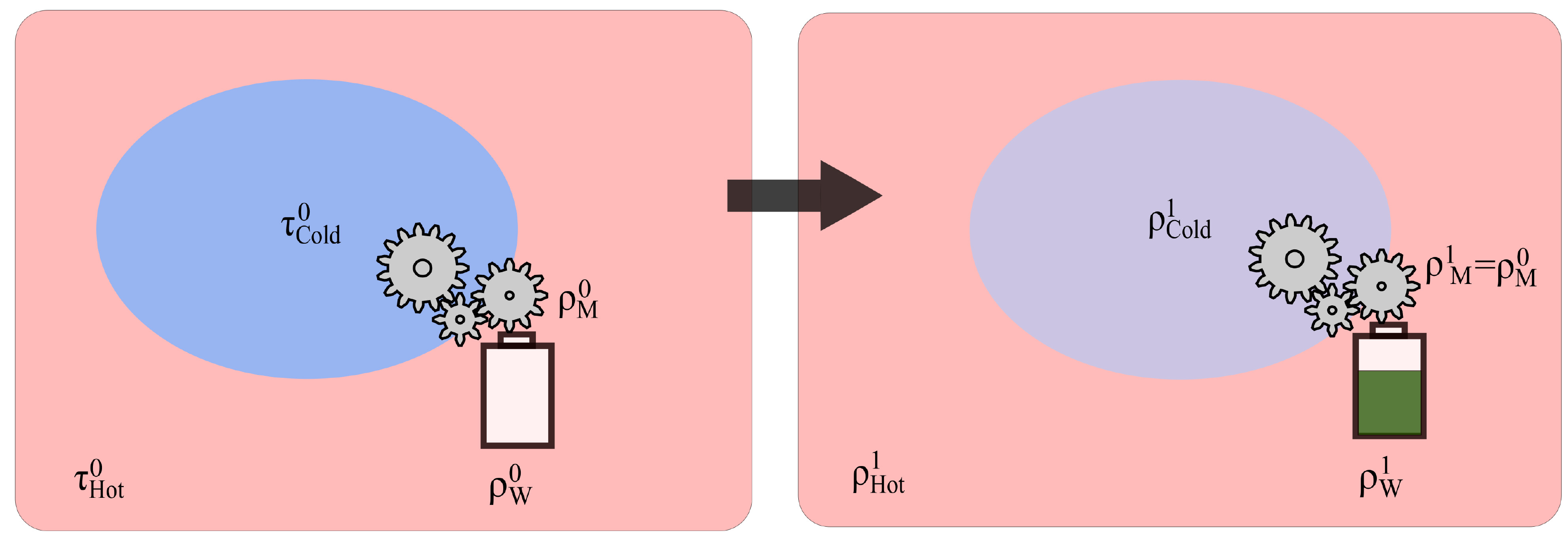}\label{fig:HE}
		\caption{A heat engine with all its basic components: 1) two baths $\tau_\cold^0=\frac{1}{\tr (e^{-\beta_\cold\hat H_\cold})}	e^{-\beta \hat H_\cold}$ and $\tau_\hot^0=\frac{1}{\tr (e^{-\beta_\hot\hat H_\hot})}	e^{-\beta \hat H_\hot}$ which are initially thermal at distinct inverse temperatures $\beta_\cold > \beta_\hot$, 2) a machine $\rho_\mach^0$ which utilizes this temperature difference to extract work, while undergoing a cyclic process, i.e. $\rho_\mach^1=\rho_\mach^0$, and 3) a battery that goes from $\rho_\batt^0\rightarrow\rho_\mach^1$ and stores the extracted energy.}
		\end{figure}
	\end{center}
	\vspace{-0.9cm}
	
	Recently, a number of schemes for QHEs have been proposed and analyzed, involving systems such as ion traps, photocells, or optomechanical systems \cite{abah2012single,nanoscaleHE,scully2010quantum,scully2011quantum,zhang2014quantum,dong2015work,kieu2004second,skrzypczyk2014work,Rossnagel325}. 
	Some of these schemes lie outside the usual definition of a heat engine (see Fig.~\ref{fig:HE}). For example, instead of the engine using a hot and cold bath, the extended quantum 
	heat engine (EQHE) has access to reservoirs which are not in a thermal state~\cite{nanoscaleHE,entenchance,scully2003extracting}. 
	In this case, EQHE with high efficiencies (even surpassing $\eta_C$) have been proposed and demonstrated.
	Nevertheless, \cite{universalityQCE2015} has pointed out that the second law is, strictly speaking, never violated for such EQHEs because one always has to invest extra work in order to create and replenish these special non-thermal reservoir states.
	In contrast, here, we consider the standard setting of a quantum heat engine, in which the baths are indeed thermal. 
	In this setting, it has been proven that although the Carnot efficiency can be approached, but it can never be surpassed. \cite{WNW15}. 
	
	QHEs are radically different from classical engines, since energy fluctuations are much more prominent due to the small number of particles involved, and many assumptions for bulk systems such as ergodicity do not hold. The laws of thermodynamics for small quantum systems are more restrictive due to finite-size effects \cite{HO-limitations,2ndlaw,tajima2014refined,WNW15,quan2014maximum}. It turns out that such laws introduce additional restrictions on the performance of QHEs~\cite{WNW15}. 
	Specifically, not all QHEs can even achieve the Carnot efficiency. The maximal achievable efficiency depends not only on the temperatures, but also on the Hamiltonian structure of the baths involved. 
	Furthermore, considering a probabilistic approach towards work extraction, \cite{verley2014unlikely} found that the achievement of Carnot efficiency is a very unlikely event, when considering energy fluctuations in the microregime.
	
	Can we design a QHE that operates between genuinely thermal reservoirs and yet achieves a high efficiency? To answer this question, several protocols have been proposed \cite{esposito2010quantum,skrzypczyk2014work} and analyzed, showing that they operate at the Carnot efficiency. 
	Crucial to these results is the definition of work. 
In these approaches, the most common approach of quantifying work is to measure the average increase in energy of an ancillary system, sometimes referred to as the battery, after a certain work extraction protocol \cite{skrzypczyk2014work,PhysRevLett.111.240401,extworkcorr,GHRdR2016,vinjanampathy2015quantum}. 
	The quality of work extracted is usually argued to be good by quantifying higher moments of the energy distribution, or by restricting the amount of entropy to be low. However, while such approaches limit the amount of heat contribution to the energy extracted, they do not completely prevent it. What's more, no justification goes into using such a definition of work. 
	A universally agreed upon definition of performing microscopic work is lacking, and it remains a constantly debated subject in the field of quantum thermodynamics \cite{aaberg2013truly,definingworkEisert,gemmer2015single,hossein2015work,vinjanampathy2015quantum}.
	This is mainly why a complete picture describing the performance limits of a QHE remains unknown. 
	
	The goal of our paper is to show that average energy increase is not an adequate definition of work for microscopic quantum systems when considering heat engines, even when 
imposing further restrictions such as a limit on entropy increase. Specifically, we demonstrate that if one allows for a definition of work that tolerates a non-negligible entropy increase in the battery, then one can in fact exceed the Carnot efficiency. Most importantly, this can already happen when 1) the cold bath only consists of 1 qubit, where we know that finite-size effects further impede the possibility of thermodynamic state transitions, and 2) without using any additional resources such as non-thermal reservoirs in EQHEs. 
The reason for being able to surpass the Carnot efficiency stems from the fact that heat contributions have ``polluted" our definition of work extraction. 
	We show that work can be divided into different categories: perfect and near-perfect work, where heat (entropy) contributions are negligible with respect to the energy gained; while imperfect work characterizes the case where heat contributions are comparable to the amount of energy gain. We find examples of extracting imperfect work where the Carnot efficiency is surpassed. This completes our picture of the understanding of work in QHEs, since we already know that by drawing perfect/near perfect work, no QHE can ever surpass the Carnot efficiency \cite{WNW15}.

	\section*{General setting of a heat engine}\label{section:1Setting}
	\subsection*{The setup}
	Let us first describe a generic quantum heat engine, which is in essence, a procedure for extracting work from two thermal baths at a temperature difference. 
	Such an engine comprises of four basic elements: two thermal baths at distinct temperatures $T_\hot$ and $T_\cold$ respectively, a machine, and a battery. 
	The machine interacts with these baths in such a way that utilizes the temperature difference between the two baths to perform work extraction. 
	The extracted work can then be transferred and stored in the battery, while the machine returns to its original state. 
	The total Hamiltonian 
	\begin{equation}\label{eq:totalH}
		\hat H_t = \hat H_\cold +\hat H_\hot + \hat H_\mach +\hat H_\batt,
	\end{equation}
	is the sum of each individual Hamiltonian, where the indices Hot, Cold, M, W represent a hot thermal bath (Hot), a cold thermal bath (Cold), a machine (M), and a battery (W) respectively. 
	Let us also consider an initial state $\rho^{0}_\total=\tau_\cold^{0}\otimes \tau_\hot^{0}\otimes \rho_\mach^{0}\otimes \rho_\batt^{0}$.
	We assume the systems were initially brought together in an uncorrelated fashion, because they have not interacted with each other beforehand.
	The state $\tau_\hot^{0}$ ($\tau_\cold^{0}$) is the initial thermal state at temperature $T_\hot$ ($T_\cold$), corresponding to the hot (cold) bath Hamiltonian $\hat H_\hot (\hat H_\cold)$, and $T_\cold<T_\hot$. 
	More generally, given any Hamiltonian $\hat H$ and temperature $T$, the thermal state is defined as $\tau=\frac{1}{\tr (e^{-\hat H/k_B T})}	e^{-\hat H/k_B T}$. 
	For notational convenience, we shall often work with inverse temperatures, defined as $\beta_h:=1/k_B T_\hot$ and $\beta_c:=1/k_B T_\cold$ where $k_B$ is the Boltzmann constant. 
	The initial machine $(\rho_\mach^0, \hat H_\mach)$ can be chosen arbitrarily, as long as its final state is preserved (and therefore the machine acts like a catalyst). 

	We adopt a thermodynamic resource theory approach \cite{brandao2013resource,HO-limitations}, allowing all unitaries $U$ on the global system such that $[U,\hat H_\total]=0$. Such a set of operations conserve the energy of the global system, which is a requirement based on the first law of thermodynamics. 
	If $(\tau_\hot^0, \hat H_\hot)$ and $(\rho_\mach^0, \hat H_\mach)$ can be arbitrarily chosen, then any such unitary $U$, $(\tau_\hot^0, \hat H_\hot)$ and $(\rho_\mach^0, \hat H_\mach)$  defines a \emph{catalytic thermal operation} \cite{2ndlaw} which one can perform on the joint state $\CW$. This implies that the cold bath is used as a non-thermal resource, relative to the hot bath. 
	By catalytic thermal operations that act on the cold bath, using the hot bath as a thermal reservoir, and the machine as a catalyst, one can extract work and store it in the battery. 

	The aim is to achieve a final \emph{reduced} state $\rho^1_\total$, such that
	\begin{align}\label{eq:rho1}
		\rho_\CMW^1 = \tr_\hot (\rho^{1}_\total)=\rho_\CW^{1}\otimes \rho_\mach^{1},\qquad\qquad
	\end{align}
	where $\rho_\mach^{1}=\rho_\mach^{0}$, and $\rho_\cold^{1}$ is the final joint state of the cold bath and battery. 
	For any bipartite state $\rho_{\rm AB}$, we use the notation of reduced states $\rho_{\rm A}:=\tr_\textup{B}(\rho_{\rm AB})$.

	Finally, we need to describe the battery such that the state transformation from $\rho^{0}_\total$ to $\rho^{1}_\total$ stores work in the battery. 
	This is done as follows: consider the battery which has a Hamiltonian (written in its diagonal form)
	$\hat H_\batt:=\sum_{i=1}^{n_\batt} E^\batt_i|E_i\ra \la E_i|_\batt$.
	For some parameter $\varepsilon\in [0,1)$, we consider the initial and final states of the battery to be
	\begin{align}
		\rho_\batt^{0}&=|E_j\ra\la E_j|_\batt\label{eq:rhobatt_int}\\
		\rho_\batt^{1}&=(1-\varepsilon)|E_k\ra\la E_k|_\batt+\varepsilon|E_j\ra\la E_j|_\batt\label{eq:rhobatt_fin}
	\end{align}
	respectively. This can be seen as a simple form of extracting work: going from a pure energy eigenstate to a higher energy eigenstate (except with some probability of failure). In \cite{WNW15}, it has been shown that a much more general form of the battery states may be allowed.
	The extracted work during a transformation $\Wext$ is then defined as the energy difference 
	\be\label{eq:Wextdef}
		\Wext:=E_k^\batt-E_j^\batt.\qquad\qquad\qquad\qquad\quad
	\ee
	where we define $E_k^\batt>E_j^\batt$ such that $\Wext>0$. 
	The parameter $\varepsilon$ corresponds to the failure probability of extracting work, usually chosen to be small. 

	To summarize, so far we have made the following minimal assumptions:
	\begin{itemize}
		\item [\textbf{ (A.1)}]Product state: 
			There are no initial correlations between the cold bath, machine and battery. 
			Initial correlations we assume do not exist, since each of the initial systems are brought independently into the process. 
			This is an advantage of our setup, since if one assumed initial correlations, one would then have to use unknown resources to generate them in the first place. 
		\item [\textbf{ (A.2)}]Perfect cyclicity: 
			The machine undergoes a cyclic process, i.e. $\rho_\mach^0=\rho_\mach^{1}$, and is also not correlated with the rest of the cold bath and battery, as described in Eq.~\eqref{eq:rho1}. This is to ensure that the machine does not get compromised in the process: since if $\rho_\mach^0$ was initially correlated with some reference system R, then by monogamy of entanglement, correlations between $\rho_\mach^1$ and $\rho_\CW^1$ would potentially destroy such correlations between the machine M with R.
		\item[\textbf{ (A.3)}]Isolated quantum system: 
			The heat engine as a whole, is isolated from and does not interact with the world. This assumption ensures that all possible resources in a work extraction process have been accounted for.
		\item[\textbf{ (A.4)}] Finite dimension: 
			The Hilbert space associated with $\rho_\total^0$ is finite dimensional but can be arbitrarily large. Moreover, the Hamiltonians $\hat H_\cold,$ $\hat H_\hot,$ $\hat H_\mach$ and $\hat H_\batt$ all have bounded pure point spectra, meaning that these Hamiltonians have eigenvalues which are bounded. This assumption comes from the resource theoretic approach of thermodynamics \cite{HO-limitations}.
	\end{itemize}
		
	\subsection*{Quantifying work and efficiency} 
	From \cite{WNW15} we know that there is an interplay between the values of $\varepsilon$ with the maximum extractable work, $\Wext$. 
	Let us first look at $\varepsilon$: this failure probability injects a certain amount of entropy into the battery's final state, therefore compromising the quality of extracted work. For an initially pure battery state, let $\Delta S$ denote the von Neumann entropy of the final battery state,
	\be\label{eq:DeltaSDef}
		\Delta S:=-\rho_\batt^1\ln\rho_\batt^1=-\varepsilon\ln \varepsilon -(1-\varepsilon)\ln (1-\varepsilon).
	\ee
	Since the probability distribution of the final battery state has its support on a two-dimensional subspace of the battery system, this definition also coincides with the binary entropy of $\varepsilon$, denoted by $h_2 (\varepsilon)$.
	
	The more entropy $\Delta S$ is created in the battery, the more disordered is the energy one extracts, i.e. the larger are the heat contributions. Ideally, we would like zero entropy; where the final state of the battery is simply a pure energy eigenstate $\ketbra{E_k}{E_k}_\batt$ with $E_k^\batt-E_j^\batt>0$. Not only then we obtain a net increase in energy, but also we have full knowledge of the final battery state. Taking another extreme example, for a fixed amount of average energy increase, $\Delta S$ is maximized when the final state of the battery is thermal. A thermal state by itself cannot then be used to obtain work; it has to be combined with other resources (for example, another heat bath at a different temperature) in order to obtain ordered work.
	
	However, the absolute value of $\Delta S$ is less important by itself; instead we want to compare it with the amount of energy extracted. Therefore, we may categorize work into the following regimes: 
	\begin{definition}(Perfect work \cite{WNW15})\label{def:perfectwork}
		An amount of work extracted $W_\textup{ext}$ is referred to as \textup{perfect work} when $\varepsilon=0$.
	\end{definition}	
	\begin{definition}(Near perfect work \cite{WNW15})\label{def:nearperfectwork}
		An amount of work extracted $W_\textup{ext}$ is referred to as \textup{near perfect work} when
		\begin{itemize}
			\item[1)] $0< \varepsilon\leq l,~$ for some fixed $l<1$ and
			\item[2)] $\displaystyle 0<\frac{\Delta S}{W_\textup{ext}}<p~$ for any $~p>0$, 
						i.e. $\displaystyle\frac{\Delta S}{W_\textup{ext}}$ is arbitrarily small.
		\end{itemize}
	When $\Wext$ is finite, items 1) and 2) are both satisfied only in the limit $\varepsilon\rightarrow 0$, if and only if $\displaystyle\lim_{\varepsilon\rightarrow 0^+} \frac{\Delta S}{W_\textup{ext}} =0.$
	\end{definition} 
	\begin{definition}(Imperfect work)\label{def:imperfectwork}
	An amount of work extracted $W_\textup{ext}$ is referred to as \textup{imperfect work} when
		\begin{itemize}
			\item[1)] $0< \varepsilon\leq l,~$ for some fixed $l<1$ and
			\item[2)] $\displaystyle \frac{\Delta S}{W_\textup{ext}}=p~$ for some value $~p>0$, 
						i.e. $\displaystyle\frac{\Delta S}{W_\textup{ext}}$ is lower bounded away from zero.
		\end{itemize}
	\end{definition}

	Let us also define the notion of a \emph{quasi-static} heat engine, which will be important in our analysis. 
	\begin{definition}(Quasi-static \cite{WNW15})\label{def:quasistatic}
		A heat engine is \textup{quasi-static} if the final state of the cold bath is a thermal state and its inverse temperature $\beta_f$ only differs infinitesimally from the initial cold bath temperature, i.e. $\beta_f=\beta_c-g$, where $0<g\ll 1$. We also refer to $g$ as the quasi-static parameter.
	\end{definition}
	
	Having fully described the QHE setup, one then asks: for what values of $\Wext$ can the transition $\rho_\total^0\rightarrow\rho_\total^1$ occur? The possibility of such a thermodynamic state transition depends on a set of conditions derived in \cite{2ndlaw}, phrased in terms of quantities called generalized free energies (see Appendix for more details). These conditions place upper bounds on the amount of work $\Wext$ extractable, and since our initial states are block-diagonal in the energy eigenbasis, these second laws are necessary and sufficient to characterize a transition.

	The efficiency of a particular heat engine is given by 
	\begin{equation}\label{eq:def_efficiency}
	\eta := \frac{\Wext}{\Delta H},
	\end{equation}
	where $\Delta H = \tr (\hat H_\hot \tau_\hot^0) - \tr (\hat H_\hot \rho_\hot^1)$. This can be simplified by noting that the total Hamiltonian in Eq.~\eqref{eq:totalH} is simply the individual sum of each system's free Hamiltonian, and therefore for any state $\rho_\total$,
	\begin{equation}
	tr (\hat H_\total \rho_\total) = \tr (\hat H_\hot \rho_\hot)+\tr (\hat H_\cold \rho_\cold)+\tr (\hat H_\batt \rho_\batt)+\tr (\hat H_\mach \rho_\mach).
	\end{equation}
	If we define the terms
	\begin{align*}
	\Delta C &= \tr (\hat H_\cold \rho_\cold^1) - \tr (\hat H_\cold \tau_\cold^0),\\
	\Delta W &= \tr (\hat H_\batt \rho_\batt^1) - \tr (\hat H_\batt \rho_\batt^0),
	\end{align*}
	then we see that since total energy is preserved in the process,
	\begin{align*}
	\tr (\hat H_\hot \tau_\hot^0)+\tr (\hat H_\cold \tau_\cold^0)+\tr (\hat H_\batt \rho_\batt^0)+\tr (\hat H_\mach \rho_\mach^0) 
	&= \tr (\hat H_\hot \rho_\hot^1)+\tr (\hat H_\cold \tau_\cold^1)+\tr (\hat H_\batt \rho_\batt^1)+\tr (\hat H_\mach \rho_\mach^1) .
	\end{align*}
	By noting that $\rho_\mach^0=\rho_\mach^1$ and rearranging terms, we have $\Delta H = \Delta C + \Delta W$. Furthermore, note that because of Eqns.~\eqref{eq:rhobatt_int} and \eqref{eq:rhobatt_fin}, we have $\Delta W = (1-\varepsilon)\Wext$. Hence, according to Eq.~\eqref{eq:def_efficiency}, we have
	\begin{equation}\label{eq:simplified_eff1}
		\eta^{-1} =  1-\varepsilon + \frac{\Delta C}{\Wext}.
	\end{equation}
	\section*{Results}	
	The main result of this paper is that: we show that Carnot efficiency can be surpassed in a single-shot setting of work extraction, even without using additional non-thermal resources. We obtain this result through deriving an analytical expression for the efficiency of a QHE in the quasi-static limit, when extracting imperfect work. 

      Consider the probability $\varepsilon$ where the final battery state is not in the state $\ketbra{E_k}{E_k}$, as according to Eq.~\eqref{eq:rhobatt_fin}. This is also what we call the failure probability of extracting work. The limit $\varepsilon\rightarrow 0$ is the focus of our analysis for several reasons. Firstly, recall that when categorizing the quality of extracted work, one is interested not only in the absolute values of entropy change in the battery, which we have denoted as $\Delta S$. Rather, this entropy change compared to the amount of extracted work $\Wext$, in other words the ratio $\frac{\Delta S}{\Wext}$ is the quantity of importance.
    For any given finite $n$ number of cold bath qubits, the amount of work extractable is finite. Extracting near perfect work means that the entropy $\Delta S$ should be negligible compared with extractable work $\Wext$, as we have also seen in Def.~\ref{def:nearperfectwork}. Since according to Eq.~\eqref{eq:DeltaSDef}, $\Delta S = h_2 (\varepsilon) \geq \varepsilon$, therefore we are concerned with the limit where $\varepsilon$ is arbitrarily small.
    On the other hand, now consider imperfect work. The quasi-static limit, i.e. $g\rightarrow 0$ is the focus of our analysis that aims to provide examples of imperfect work extraction. In the quasi-static limit, since the cold bath changes only by an infinitesimal amount, therefore the amount of work extractable $\Wext$ is also infinitesimally small. 
    For most cases of imperfect work (when the ratio of $\frac{\Delta S}{\Wext}$ is finite) we know that $\Delta S$ is vanishingly small, and therefore so is the quantity $\varepsilon$.
    
        	In \cite{WNW15}, it has already been shown that perfect work is never achievable, while considering near perfect work allows us to sometimes achieve arbitrarily near to Carnot efficiency, but not always. 
	Therefore, our results, when combining with \cite{WNW15} provide the full range of possible limits for $\frac{\Delta S}{\Wext}$, with the corresponding findings about the maximum achievable efficiency, which we summarize in Table \ref{table:1}. 
        \begin{table*}[ht]
    \begin{tabular}{|*{3}{c|}}
    \hline
     \multicolumn{2}{|c|}{\textbf{Type}} & \textbf{Maximum efficiency} \\[5pt] 
    \hline
    \rule{0pt}{3ex} 
     Perfect work &$\varepsilon=0$ \cite{WNW15} & ~Work extraction for any $\Wext >0$ is not possible. \\  [3pt] 
    \hline 
    \rule{0pt}{4.2ex} 
     ~~ Near perfect work \qquad & $\displaystyle\lim_{\varepsilon\rightarrow 0}\frac{\Delta S}{\Wext}=0$~\cite{WNW15} & \multirow{4}{*}{\specialcell{$\eta_C$ is the theoretical maximum, and can only be approached uniquely\\in the quasi-static limit. However, $\eta_C$ can be approached only if certain~\\conditions on the bath Hamiltonian are met. Otherwise, the maximum\\attainable efficiency is strictly upper bounded away from $\eta_C$.}} \\  [8pt] 
    \hhline{--~}
    \rule{0pt}{1ex} 
     &  & \\ [3pt] 
        \hhline{~~}
    \rule{0pt}{3.8ex} 
  \multirow{2}{*}{\specialcell{Imperfect work\\(this paper)}} & $\quad\displaystyle\lim_{\varepsilon\rightarrow 0}\frac{\Delta S}{\Wext}= p\in (0,\infty)\quad$ & \\ [10pt] 
    \hhline{~--}
    \rule{0pt}{3ex} 
    & $\infty$ & ~Unknown, however examples of exceeding CE can be found. \\ [3pt] 
    \hline 
    \end{tabular}
    \caption{Different regimes of work corresponding to different limits of the ratio $\displaystyle\lim_{\varepsilon\rightarrow 0} \frac{\Delta S}{\Wext}$.}\label{table:1}
    \end{table*}
    
	    Theorem \ref{thm:mainresult} formally states our main result. This theorem establishes a simplification of the efficiency of a quasi-static heat engine, given a cold bath consisting of $n$ identical qubits, each with energy gap $E$. In this theorem, we consider a special case where the failure probability $\varepsilon\propto g$ is proportional to the quasi-static parameter $g$ (see Def.~\ref{def:quasistatic}), and evaluate the efficiency in the limit $g\rightarrow 0$. In the appendix, we show that this corresponds to extracting imperfect work, in particular, $\displaystyle\lim_{\varepsilon\rightarrow 0}\frac{\Delta S}{\Wext}=\infty$. For such a case, we show that whenever 
$E < \frac{1}{2(\beta_c-\beta_h)}$,
    then for some parameter $\alpha^*$, we can choose the proportionality constant $c (\alpha^*) = \frac{\varepsilon}{g}$ such that the corresponding efficiency of such a heat engine is given by a simple analytical expression.
    Therefore, by numerically evaluating such an expression for different parameters $\beta_c, \beta_h, E, n, \alpha^*$ etc, one can find examples of surpassing the Carnot efficiency. 
	\begin{theorem}[\textbf{Main Result}]\label{thm:mainresult}
		Consider a quasi-static heat engine with a cold bath consisting of $n$ identical qubits with energy gap $E>0$. Given the inverse temperatures of the hot and cold bath $\beta_h,\beta_c>0$ respectively, and for $\alpha\in (1,\infty)$ define the functions 
		\begin{equation}
		B_\alpha= \frac{E}{1+e^{\beta_c E}}\cdot \frac{e^{(\beta_h+\alpha\beta_c)E}-e^{(\beta_c+\alpha\beta_h)E}}{e^{\alpha\beta_h E}+e^{(\beta_h+\alpha\beta_c)E}}
		\end{equation}
		and $B_\alpha' = \diff{B_\alpha}{\alpha}$ being the first derivative of $B_\alpha$ according to $\alpha$.
		If the energy gap of the qubits satisfy
		\begin{equation}\label{eq:mainresult_Egap}
		 0<E<\frac{1}{2(\beta_c-\beta_h)},
		 \end{equation} 
		 then there exists an $\alpha^*\in(1,2)$ such that the failure probability 
		 \begin{equation}\label{eq:mainresult_eps}
		 \varepsilon = g\cdot n [\alpst(\alpst-1)B_{\alpst}' - B_{\alpst}] >0,
		 \end{equation}
		 and the inverse efficiency (Eq.~\eqref{eq:simplified_eff1}) of the described heat engine is given by 
		 \begin{equation}\label{eq:mainresult_eff}
		 \eta^{-1} = 1 + \frac{\beta_h}{\beta_c-\beta_h}\frac{1}{{\alpha^*}^2}\frac{B_1'}{B_{\alpha^*}'}.
		 \end{equation}
	\end{theorem}	
	
	The approach taken to prove Theorem \ref{thm:mainresult} is further explained in the Methods section.
	
		\begin{figure}[h!]      
		\centering 
        \begin{minipage}{0.32\textwidth}
                \includegraphics[width=\textwidth]{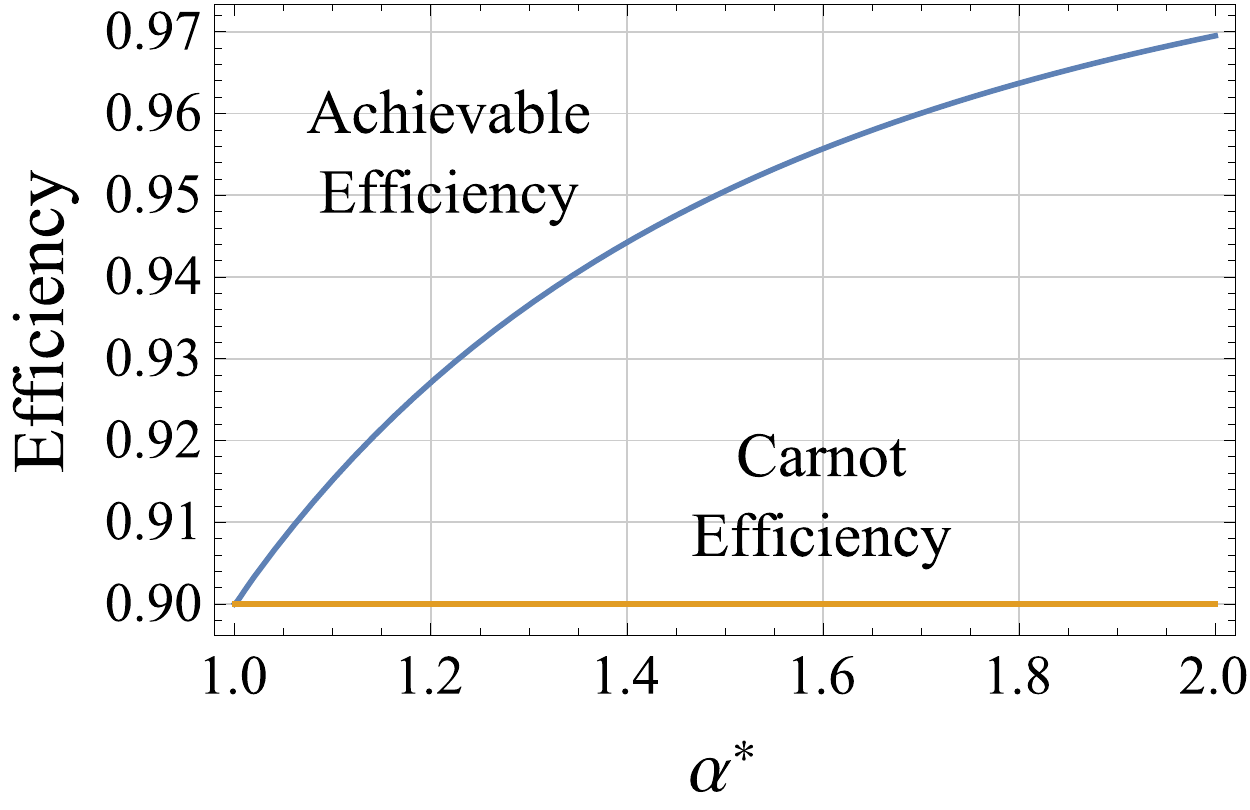}
                \caption{Achievable efficiency versus Carnot efficiency with respect to $\alpha^*\in(1,2)$, $\beta_h = 1$, $\beta_c = 10$ and $E = \frac{0.4}{\beta_c-\beta_h}$.}
                \label{fig:a}
        \end{minipage}\hspace{0.1cm}%
        ~ 
        \begin{minipage}{0.32\textwidth}
                \includegraphics[width=\textwidth]{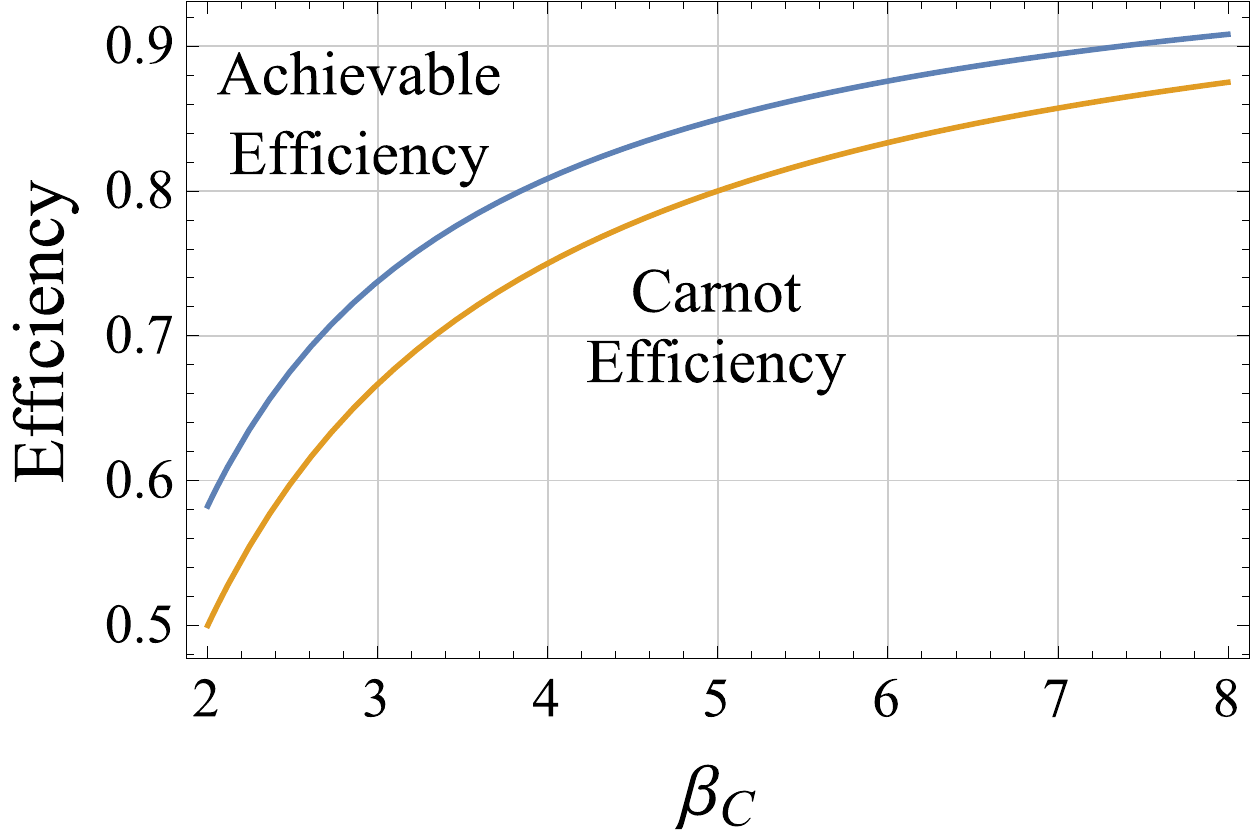}
                \caption{Achievable efficiency versus Carnot efficiency with respect to $\beta_c$, with $\alpha^*=1.2$, $\beta_h=1$, $E=\frac{0.4}{\beta_c-\beta_h}$.}
                \label{fig:b}
        \end{minipage}\hspace{0.1cm}
        ~ 
        \begin{minipage}{0.32\textwidth}
                \includegraphics[width=\textwidth]{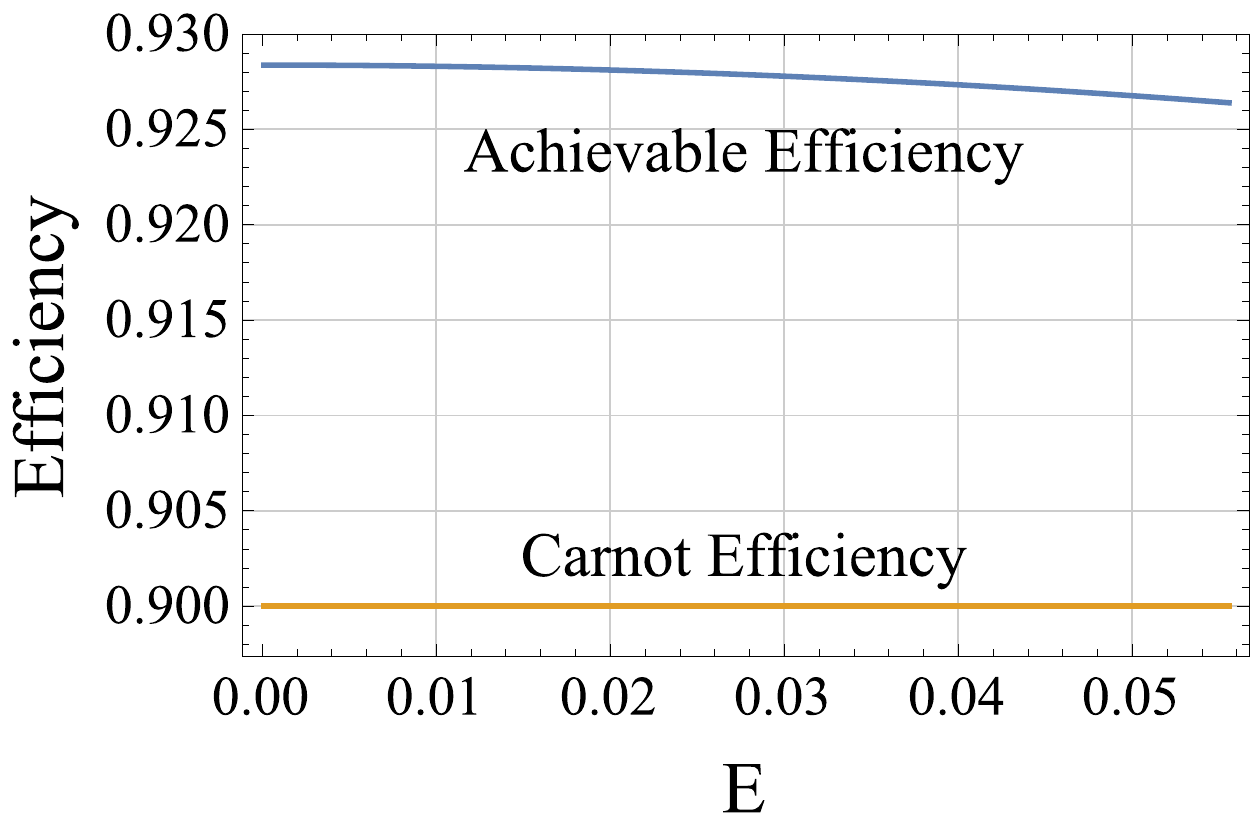}
                \caption{Achievable efficiency versus Carnot efficiency with respect to $E$, with $\alpha^*=1.2$, $\beta_h=1$, $\beta_c=10$, $E=\frac{0.4}{\beta_c-\beta_h}$.}
                \label{fig:c}
        \end{minipage}
        \end{figure}
        
        We plot, in Figures \ref{fig:a}-\ref{fig:c} the comparison between Carnot efficiency and the efficiency achievable according to Theorem \ref{thm:mainresult}. In all these plots we observe that Carnot efficiency is always surpassed. 
        It is worth noting that Eq.~\eqref{eq:mainresult_Egap} is in the regime where if one considers drawing near perfect work, it is always possible to achieve arbitrarily close to Carnot efficiency according to \cite{WNW15}.Therefore, the blue curve never falls below the yellow line. The improvement in efficiency happens most when the parameter $\alpst$ is adjusted, since this is the parameter that determines how quickly the ratio $\frac{\Delta S}{\Wext}\rightarrow\infty$ in the quasi-static limit.
        
Given that in Table \ref{table:1}, the case of $p\in(0,\infty)$ also corresponds to imperfect work, one might wonder if Carnot efficiency can also be surpassed in this regime. We show that this is not possible. However, if only the standard free energy is responsible for determining state transitions, then Carnot efficiency again might be surpassed.         
        

	\section*{Methods}
	There are several steps taken in order to achieve the proof of Theorem \ref{thm:mainresult}, which we outline in this section. For details, the reader is referred to Corollary \ref{cor:HamiltonianConditionsForAlphaStarMinimization} and its proof in the Appendix, which directly implies Theorem \ref{thm:mainresult}. 
	
	Theorem \ref{thm:mainresult} is obtained by considering a cold bath of $n$-identical qubits, and calculating the ratio of extractable work $\Wext$ against $\Delta C$ in the quasi-static limit, i.e. $g\rightarrow 0^+$. Then, by using Eq.~\eqref{eq:simplified_eff1}, one can evaluate the efficiency. The main difficulty lies in evaluating $\Wext$, the amount of extractable work. This quantity represents the maximum value of the battery's energy gap, such that a transition $\tau_{\beta_\cold}\ot\rho_\batt^0\rightarrow\rho_\cold^1\ot\rho_\batt^1$ is possible according to the generalized second laws described in Appendix \ref{sec:2secondlaws}.
	Applying the generalized second laws, we can calculate $\Wext$, which is given by a minimization problem over the continuous range of a real-valued variable $\alpha > 0$,
	\begin{equation}\label{eq:Wextinf}
 			\Wext =\inf_{\alpha > 0} W_\alpha,
 			\end{equation}
 			where
		 \begin{align}\label{eq:Walp_identicalqubits_quasistatic1}		 
			W_\alpha &=
			\begin{cases}
				\frac{1}{\beta_h (\alpha-1)} \left[\alpha n g B_\alpha - \varepsilon^\alpha +\alpha\varepsilon\right] + \bo(g^2)+\bo(\varepsilon^{2\alpha})+\bo(g\varepsilon^\alpha)+\bo(\varepsilon^2) &\mbox{if } \alpha\in(0,\infty)\backslash\lbrace 1 \rbrace,\\
				\left[\displaystyle\lim_{\alpha\rightarrow 1^+}\frac{1}{\beta_h (\alpha-1)} \left(\alpha n g B_\alpha - \varepsilon^\alpha +\alpha\varepsilon\right)\right] + \bo (\varepsilon g)+ \bo (\varepsilon^2 \ln\varepsilon) + \bo (\varepsilon^2)+\bo (g^2) &\mbox{if } \alpha=1.
			\end{cases}
		\end{align}
	Therefore, the difficulty of evaluating the efficiency lies 
	in performing the optimization of $W_\alpha$ over $\alpha\in(0,\infty)$, which is neither monotonic nor convex. However, by manipulating our freedom of choosing $\varepsilon$, we show that in certain parameter regimes of $\beta_c, \beta_h, $ and $E$, one can evaluate a simple, analytical expression for $\Wext$. The steps taken are outlined as follows, while all the technical lemmas are proven in the Appendix:
	\begin{enumerate}
	\item We start by choosing the failure probability to be $\varepsilon = \varepsilon_1\cdot g$, where $\varepsilon_1$ is independent of the quasi-static parameter $g$.
	\item Starting out from the expression for extractable work given in Lemma \ref{lem:quasiHE_nidqubits}, we prove that in the quasi-static limit, the regime $\alpha\in (0,1)$ need not be considered in the optimization. This is proven in Lemma \ref{lem:inf_not_below_one}.
	\item We show that the function $W_\alpha$ which we desire to minimize has at most one unique local minima. To do so, we establish technical Lemmas \ref{lem:falpha_is_concave}, \ref{lem:notmorethantwostpt} and \ref{lem:Winfty_approachedfromabove}, in order to arrive at Lemma \ref{lem:onlyonelocalmin}.
	\item We show that $\varepsilon_1$ can be chosen such that $\varepsilon >0$ (Lemma \ref{lem:condition_positive_eps1}), and that we can choose it so that we know that a particular $\alpst \in (1,2)$ corresponds to a local stationary point (Lemma \ref{lem:alpstisstationarypoint}) and specifically a local minima (Lemma \ref{lem:seconddiff>0}). Since we have established Item 3, this implies that we have identified a unique local minima.
	\item We show that under certain conditions, $W_{\alpst} <W_\infty$. This implies that $W_{\alpst}$ corresponds to the global minima which we desire to evaluate. 
	\item The conditions for Items 3-5 are summarized in Corollary~\ref{cor:HamiltonianConditionsForAlphaStarMinimization}, where one can now, by choosing the parameter $\alpst$ directly evaluate $\Wext$ analytically, and therefore use 
		\begin{equation}\label{eq:simplified_eff}
		\eta^{-1} =  1-\varepsilon + \frac{\Delta C}{\Wext}
	\end{equation}
	to calculate the efficiency. The calculation of $\Delta C$ is straightforward once $\rho_\cold^0, \rho_\cold^1$ are fixed, and for the quasi-static limit, we expand $\Delta C$ in terms of the quasi-static parameter $g$. 
	\end{enumerate}

    One can ask whether it is possible to always exceed Carnot efficiency when imperfect work is drawn. For example, observing in Table \ref{table:1} that the case of $p\in(0,\infty)$ also corresponds to imperfect work, one might wonder if a similar result of exceeding Carnot efficiency can be achieved in the regime where $\frac{\Delta S}{\Wext}\rightarrow p$ instead of $\frac{\Delta S}{\Wext}\rightarrow\infty$ (as in the case where $\varepsilon\propto g$). We show in Appendix \ref{subsec:entropy_comparable} that this is not possible, i.e. Carnot efficiency remains the theoretical maximum when the ratio $\frac{\Delta S}{\Wext}$ remains finite in the quasi-static limit. It is interesting to note that, if only the standard free energy is responsible for determining state transitions, then Carnot efficiency again might be exceeded. In conclusion, in the regime where $p$ is finite, the reason that one cannot exceed Carnot efficiency stems from the fact that there exists a continuous family of generalized free energies in the quantum microregime (see Appendix \ref{sec:2secondlaws}).

	\section*{Discussions and Conclusion}
	Why is it important to distinguish between work and heat? Suppose we have two batteries $A_1$ and $A_2$, each containing the same amount of average energy. However, $A_1$ is in a pure, defined energy eigenstate; while $A_2$ is simply a thermal state corresponding to a particular temperature $T_2$. Firstly, note that there is an irreversibility via catalytic thermal operations for these two batteries: the transition $A_1\rightarrow A_2$ might be possible, but certainly $A_2\nrightarrow A_1$, since the free energy of $A_1$ is higher than of $A_2$. This makes $A_1$ a more valuable resource compared to $A_2$. Indeed, if we further consider the environment to be of temperature $T_2$, then having $A_2$ is completely useless: it is passive compared to the environment and cannot be used as a resource to enable more state transitions. On the other hand, $A_1$ can be useful in terms of enabling state transitions. Even more crucially, the full amount of energy contained in $A_1$ can be transferred out, because we have full knowledge of the quantum state.
	
	 Indeed, for the case of extracting imperfect work, and in particular for the choice of $\varepsilon$ proportional to $g$, heat contributions are dominant. This is because in such an example, the average energy in the battery increases, its \emph{free energy} actually decreases. This can be seen because by using Eqns.~\eqref{eq:rhobatt_int}, \eqref{eq:rhobatt_fin} and \eqref{eq:DeltaSDef}, the free energy difference can be written as
		 $\Delta F = (1-\varepsilon) \Wext - \beta^{-1} \Delta S$,
	 and when $\varepsilon\propto g$ in the quasi-static limit, $\Delta S$ is much larger than $\Wext$.
	 This indicates that the free energy difference, instead of average energy difference in the battery would serve as a more accurate quantifier of work. Indeed, by adopting an operational approach towards this problem, \cite{definingworkEisert} has also identified the free energy to be a potentially suitable quantifier. However, also note that for large but finite values of $p$ in Table \ref{table:1}, the free energy of the battery might also decrease in the process; but Carnot efficiency still cannot be surpassed in this regime.
	 
	 Our result therefore serves as a note of caution when it comes to analyzing the performance of heat engines, that quantifying microscopic work simply by the average energy increase in the battery does not adequately account for heat contribution in the work extraction process. Therefore, this might lead to the possibility of surpassing the Carnot efficiency, despite finite-size effects, even in the absence of uniquely quantum resources such as entanglement. For example, the work extraction protocol proposed in \cite{skrzypczyk2014work} indeed corresponds to $\frac{\Delta S}{\Wext}\rightarrow\infty$, when the intial battery state is a pure energy eigenstate. With each step in the protocol, an infinitesimal amount of energy is extracted, while a finite amount of entropy is injected into the battery. 
	This reminds us that work and heat, although both may contribute to an energy gain, are distinctively different in quality (i.e. orderliness). Therefore, when considering small QHEs, it is not only important to propose schemes that extract energy on average, but also ensure that work is gained, rather than heat.
	
\section*{Acknowledgements}
NHYN and SW acknowledge support from STW,
Netherlands, an ERC Starting Grant and an NWO VIDI
Grant. MPW acknowledges support from the Engineering and Physical Sciences Research Council of the United Kingdom.


\newpage
\begin{center}
\textbf{Appendix}
\end{center}
\setcounter{section}{0}
\renewcommand\thesection{\Alph{section}}
\renewcommand\thesubsection{\arabic{subsection}}
	This appendix contains the technical material used and developed in order to prove the results of this paper. In Section \ref{sec:2secondlaws}, we introduce the main tool, namely the generalized second laws that govern a state transition for small quantum systems. Section \ref{subsec:mainresult}  contains a summary of our main result and a proof sketch. Lastly, we list all the results adapted from \cite{WNW15} in Section \ref{subsec:tools}, while the technical lemmas developed in this paper are collected in \ref{subsub:techlemma}.

	\section{Second laws: the conditions for thermodynamical state transitions}\label{sec:2secondlaws}

		Macroscopic thermodynamics says that for a system undergoing heat exchange with a thermal bath (at inverse temperature $\beta$), the Helmholtz free energy
		\begin{equation}
			F(\rho):=\la \hat H \ra_\rho-\frac{1}{\beta}S(\rho),
		\end{equation}
		is necessarily non-increasing. For macroscopic systems, this also constitutes a sufficient condition: whenever the free energy does not increase, we know that a state transition is possible.
		
		However, in the microscopic quantum regime, where only a few quantum particles are involved, it has been shown that macroscopic thermodynamics is not a complete description of thermodynamical transitions. 
		More precisely, not only the Helmholtz free energy, but a whole other family of generalized free energies have to decrease during a state transition. This places further constraints on whether a particular transition is allowed.
		In particular, if the final target state $\rho_\CW^1$ is diagonal in the energy eigenbasis, these laws also give necessary and sufficient conditions for the possibility of a transition $\rho_\CW^0\rightarrow\rho_\CW^1$ via catalytic thermal operations.

		We can apply these second laws to our scenario by associating the catalyst with $\rho^0_\mach$, and considering the heat engine state transition $\rho^0_\batt\otimes\tau_\cold^0 \rightarrow \rho_\CW^1$. 
		Since we start with $\rho^0_\batt\otimes\tau_\cold^0$ which is diagonal in the energy eigenbasis, and since catalytic thermal operations do not create coherences between energy levels, the final state $\rho_\CW^1$ is also diagonal in the energy eigenbasis. Hence, the transition from $\rho^0_\batt\otimes\tau_\cold^0 \rightarrow \rho^1_\batt\otimes\rho_\cold^1$ is possible via catalytic thermal operations iff $\forall \alpha\geq 0$ \cite{2ndlaw},
		\begin{align}\label{eq:2ndlaws_diag}
			F_\alpha(\tau_\cold^0\otimes\rho^0_\batt,\tau_\CW^h)\geq F_\alpha(\rho_\cold^1\otimes\rho^1_\batt,\tau_\CW^h),
		\end{align}
		where $\tau_\CW^h$ is the thermal state of the system at temperature $T_\hot$ of the surrounding bath. 
		The quantity $F_\alpha(\rho,\sigma)$ for $\alpha\geq 0$ corresponds to a family of free energies defined in \cite{2ndlaw}, which can be written in the form
		\begin{align}\label{eq:generalfreeenergy}
			F_\alpha(\rho,\tau_{\beta_h})=\frac{1}{\beta_h} \left[D_\alpha(\rho\|\tau)-\ln Z_{\beta_h}\right],
		\end{align}
 		where $D_\alpha(\rho\|\tau)$ are known as $\alpha$-R{\'e}nyi divergences. Sometimes we will use the short hand $F_\infty:=\lim_{\alpha\rightarrow \infty}F_\alpha$. On occasion, we will refer to a particular transition as being possible/impossible according to the $F_\alpha$ free energy constraint. By this, we mean that for that particular value of $\alpha$ and transition, Eq. \eqref{eq:2ndlaws_diag} is satisfied/not satisfied. The $\alpha$-R{\'e}nyi divergences can be defined for arbitrary quantum states, giving us necessary (but insufficient) second laws for state transitions \cite{2ndlaw,LJR2015description}. However, since we are analyzing states which are diagonal in the same eigenbasis (namely the energy eigenbasis), these laws are both neccesary and sufficient. Also, the R{\'e}nyi divergences can be simplified to
		\begin{align}\label{w no ep}
			D_\alpha(\rho\|\tau)=\frac{1}{\alpha-1}\ln \sum_i p_i^\alpha q_i^{1-\alpha},
		\end{align}
		where $p_i,$ $q_i$ are the eigenvalues of $\rho$ and the state $\tau$. The cases $\alpha=0$ and $\alpha\rightarrow1$ are defined by continuity, namely
		\begin{align}\label{eq:reyi in limits}
			D_0(\rho\|\tau)&=\lim_{\alpha\rightarrow 0^+}D_\alpha(\rho\|\tau)=-\ln \sum_{i:p_i\neq 0}q_i,
			\qquad D_1(\rho\|\tau)=\lim_{\alpha\rightarrow 1}D_\alpha(\rho\|\tau)=\sum_{i}p_i\ln \frac{p_i}{q_i},
		\end{align}
		and we also define $D_\infty$ as
		\begin{align}
			D_\infty(\rho\|\tau)&=\lim_{\alpha\rightarrow \infty^+}D_\alpha(\rho\|\tau)=\ln \max_{i}\frac{p_i}{q_i}.
		\end{align}
		The quantity $D_1(\rho\|\tau)$ is also known as the \emph{relative entropy}, while it can be checked that $F_1 (\rho,\tau)$ coincides with the Helmholtz free energy.
		We will often use the convention $D(\rho\|\tau), F(\rho,\tau)$ in place of $D_1(\rho\|\tau)$ and $F_1(\rho,\tau)$.

	\section{Results}	
	In \cite{WNW15} it has been shown that for a heat engine to extract any positive amount of work at all, $\varepsilon >0$ has to be true. Therefore, perfect work can never be drawn. 
	Also, in \cite{WNW15} the regime of near perfect work was analyzed. There, it was found that the maximum efficiency can never exceed the Carnot efficiency. 	
	
	In this paper, we develop an example of a heat engine which extracts imperfect work. In Section \ref{subsec:mainresult}, we show (our main result) how to find examples where Carnot efficiency is surpassed. More specifically, this occurs in the quasi-static limit where $\frac{\Delta S}{\Wext}\rightarrow\infty$. 
	In Section \ref{subsec:entropy_comparable} we analyze the regime where $\frac{\Delta S}{\Wext}\rightarrow p$, with $0<p<\infty$. We find that in this regime, 
 according to the generalized second laws, Carnot efficiency cannot be surpassed.

		\subsection{Main Result: An example of drawing imperfect work surpassing the Carnot efficiency}\label{subsec:mainresult}
		Our main result is stated in Theorem \ref{thm:mainresult}. Here, we present Corollary \ref{cor:HamiltonianConditionsForAlphaStarMinimization}, a more detailed version of Theorem \ref{thm:mainresult} with its proof, which is built upon all the technical lemmas derived in Section \ref{subsec:technical}.

		\begin{corollary}\label{cor:HamiltonianConditionsForAlphaStarMinimization}
		Consider a quasi-static heat engine with a cold bath consisting of $n$ identical qubits with energy gap $E$. 
		Given the inverse temperatures of the hot and cold bath $\beta_h,\beta_c>0$ respectively, and for $\alpha\in (1,\infty)$ define the functions 
		\begin{equation}
		B_\alpha= \frac{E}{1+e^{\beta_c E}}\cdot \frac{e^{(\beta_h+\alpha\beta_c)E}-e^{(\beta_c+\alpha\beta_h)E}}{e^{\alpha\beta_h E}+e^{(\beta_h+\alpha\beta_c)E}}
		\end{equation}
		and $B_\alpha' = \diff{B_\alpha}{\alpha}$ being the first derivative of $B_\alpha$ according to $\alpha$.
		If the energy gap of the qubits satisfies
		\begin{equation}\label{eq:finalcond}
		 E<\frac{1}{2(\beta_c-\beta_h)},
		 \end{equation} 
		 then there exists an $\alpha^*\in(1,2)$ such that the following holds:
		\begin{enumerate}
		\item The failure probability of the heat engine, can be chosen as $ \varepsilon = g\cdot n [\alpst(\alpst-1)B_{\alpst}' - B_{\alpst}] >0$.
		\item The amount of extractable work is $\Wext = W_{\alpst}$, given by Eq.~\eqref{eq:Walp_identicalqubits_quasistatic}.
		\item The (inverse) efficiency of the described heat engine is given by 
		 $\eta^{-1} = 1 + \frac{\beta_h}{\beta_c-\beta_h}\frac{1}{{\alpha^*}^2}\frac{B_1'}{B_{\alpha^*}'} + \bo(g^{\alpha^*-1}).$
		\end{enumerate}				 
		\end{corollary}
		\begin{proof}
		Since $\frac{1+e^{\beta_c E}}{e^{\beta_c E}-1}>1$, if Eq.~\eqref{eq:finalcond} holds, then Eq.~\eqref{eq:condition_positive_eps1} holds. Therefore Item 1 is a direct result of Lemma \ref{lem:condition_positive_eps1}. 
		
		Item 2 concerns the quantity $\Wext$ for the quasi-static heat engine, given by Eq.~\eqref{eq:Wextinf} and~\eqref{eq:Walp_identicalqubits_quasistatic1}. If one chooses $\alpst\in (1,2)$ and that Eq.~\eqref{eq:finalcond} holds, then Eq.~\eqref{eq:lemma8} holds as well, and so Lemma \ref{lem:seconddiff>0} and Lemma \ref{lem:Walpst<Winfty}. Therefore, Item 2 is true because 
		\begin{itemize}
		\item By Lemma \ref{lem:onlyonelocalmin} we know $W_\alpha$ does not have more than one distinct local minima.
		\item By Lemma \ref{lem:alpstisstationarypoint} and \ref{lem:seconddiff>0}, $W_{\alpst}$ is a unique local minima.
		\item By Lemma \ref{lem:Walpst<Winfty}, $W_{\alpst}<W_\infty$. Therefore, $W_{\alpst}$ is the global minima.
		\end{itemize}
		
		Therefore, finally, for the fixed parameters $n\in\mathbb{Z}^+, E\in\mathbb{R}, \alpha^*\in(1,2)$, we can evaluate the efficiency of our quasi-static heat engine for a cold bath comprising of identical qubits. This can be done by evaluating Eq.~\eqref{eq:simplified_eff} for our heat engine:
		\begin{align}
		\eta^{-1} = 1-\varepsilon + \frac{\Delta C}{W_\textup{ext}}.
		\end{align}
		The term $\varepsilon = \varepsilon_1\cdot g = \bo(g)$, where $\varepsilon_1 = n [\alpha^*(\alpha^*-1)B_{\alpha^*}' - B_{\alpha^*}]$ is a finite constant. Therefore we know $\displaystyle\lim_{g\rightarrow 0^+} \varepsilon =0$. 
		On the other hand, we have
		\begin{align}
		\Wext 
		&= W_{\alpha^*} \\
		&= \frac{1}{\beta_h (\alpha^*-1)} \left[\alpha^* n g B_{\alpha^*} - \varepsilon^{\alpha^*} +\alpha^* \varepsilon\right] + \bo(g^2)+\bo(\varepsilon^{2\alpha})+\bo(g\varepsilon^\alpha)+\bo(\varepsilon^2) \\
		& = \frac{\alpha^* (nB_{\alpha^*}+\varepsilon_1)}{\beta_h (\alpha^*-1)}g +\bo( g^{\alpha^*}) + \bo(g^2)+\bo(\varepsilon^{2\alpha})+\bo(g\varepsilon^\alpha)+\bo(\varepsilon^2)\\
		& = \frac{n{\alpha^*}^2B_{\alpha^*}'}{\beta_h }g +\bo( g^{\alpha^*}) + \bo(g^2)+\bo(\varepsilon^{2\alpha})+\bo(g\varepsilon^\alpha)+\bo(\varepsilon^2).\label{eq:final_expression_Walpst}
		\end{align}
		Since we choose $\varepsilon \propto g$ according to Item 1, and since $\alpha^*\in (1,2)$, the next leading order term in Eq.~\eqref{eq:final_expression_Walpst} is $\bo( g^{\alpha^*})$, therefore
		\begin{equation}\label{eq:final_expression_Walpst1}
		\Wext = \frac{n{\alpha^*}^2B_{\alpha^*}'}{\beta_h }g +\bo( g^{\alpha^*}).
		\end{equation}
		This tells us that $\Wext$ is a function of $g$ that vanishes as $g\rightarrow 0$. Also, from Lemma \ref{lem:DeltaC} we know the expression for $\Delta C$, which also vanishes with $g$.
		Therefore, combining expressions we have for $\varepsilon, \Delta C,$ in Eq.~\eqref{eq:DeltaCexpanded} and  $\Wext$ in Eq.~\eqref{eq:final_expression_Walpst1}, we have Item 3, i.e.
		\begin{align}
		\eta^{-1} &= 1+ \frac{\Delta C}{\Wext}\\
		&= 1 + \frac{\beta_h}{\beta_c-\beta_h}\frac{1}{{\alpha^*}^2}\frac{B_1'}{B_{\alpha^*}'} + \bo(g^{\alpha^*-1}),
		\end{align}
		where in the quasi-static limit ($g\rightarrow 0$), the order term vanishes. This concludes the proof.
		\end{proof}
	
		\begin{figure}[h!]      
		\centering 
        \begin{minipage}{0.315\textwidth}
                \includegraphics[width=\textwidth]{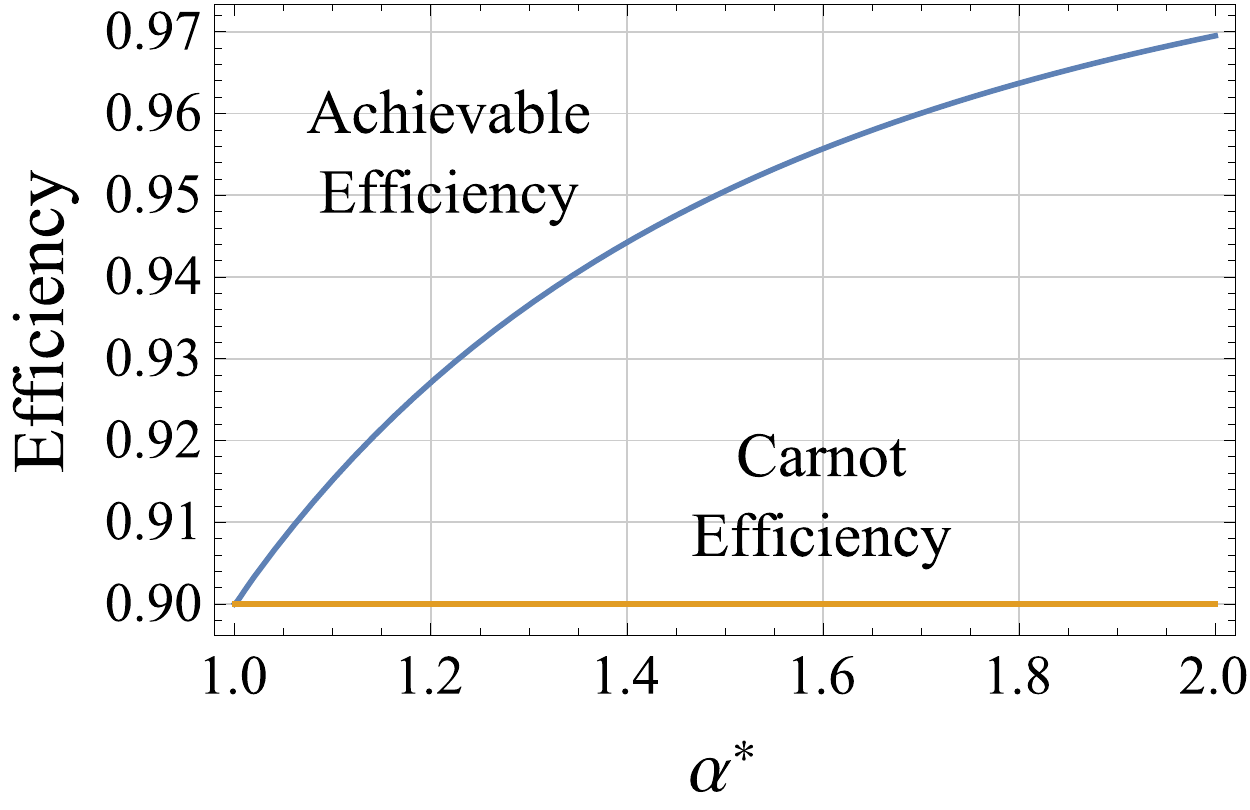}
                \caption{Achievable efficiency versus Carnot efficiency with respect to $\alpha^*\in(1,2)$, $\beta_h = 1$, $\beta_c = 10$ and $E = \frac{0.4}{\beta_c-\beta_h}$.}
                \label{fig:aa}
        \end{minipage}\hspace{0.1cm}%
        ~ 
        \begin{minipage}{0.315\textwidth}
                \includegraphics[width=\textwidth]{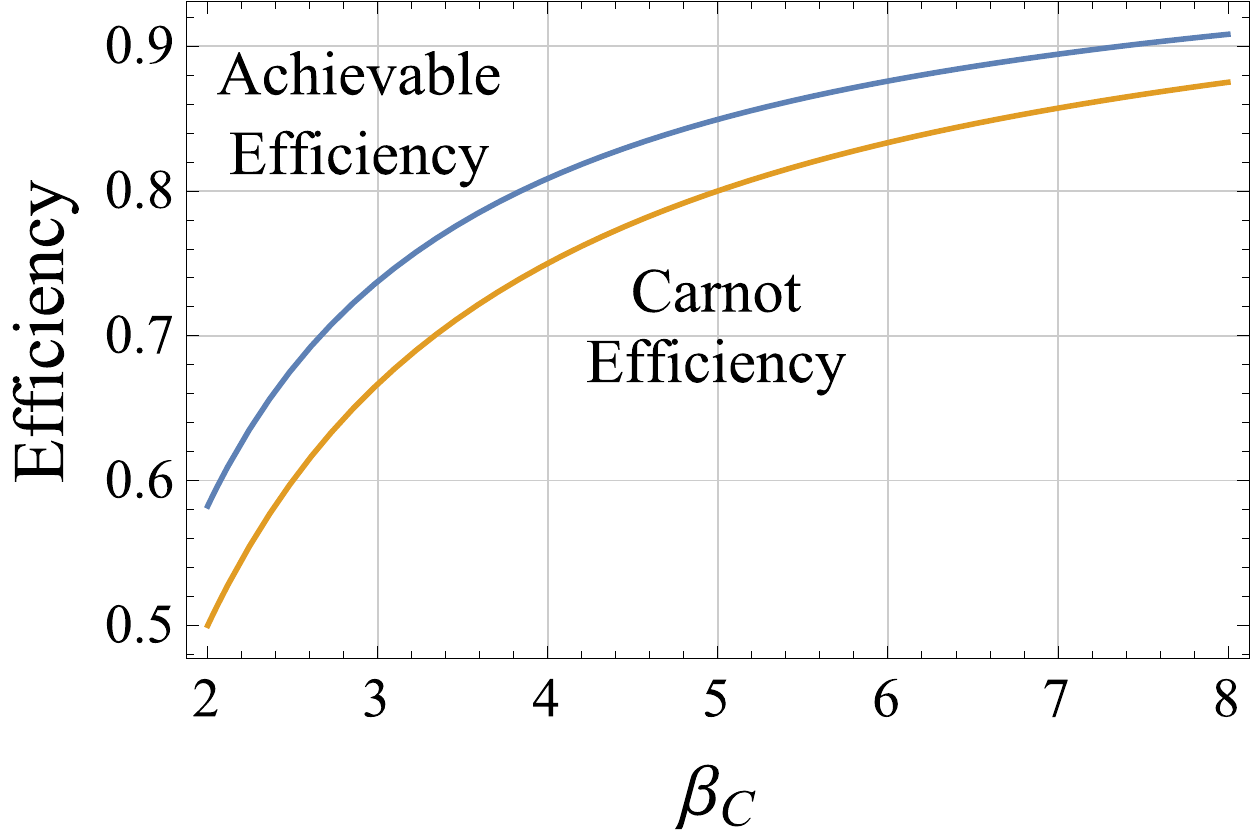}
                \caption{Achievable efficiency versus Carnot efficiency with respect to $\beta_c$, with $\alpha^*=1.2$, $\beta_h=1$, $E=\frac{0.4}{\beta_c-\beta_h}$.}
                \label{fig:bb}
        \end{minipage}\hspace{0.1cm}
        ~ 
        \begin{minipage}{0.315\textwidth}
                \includegraphics[width=\textwidth]{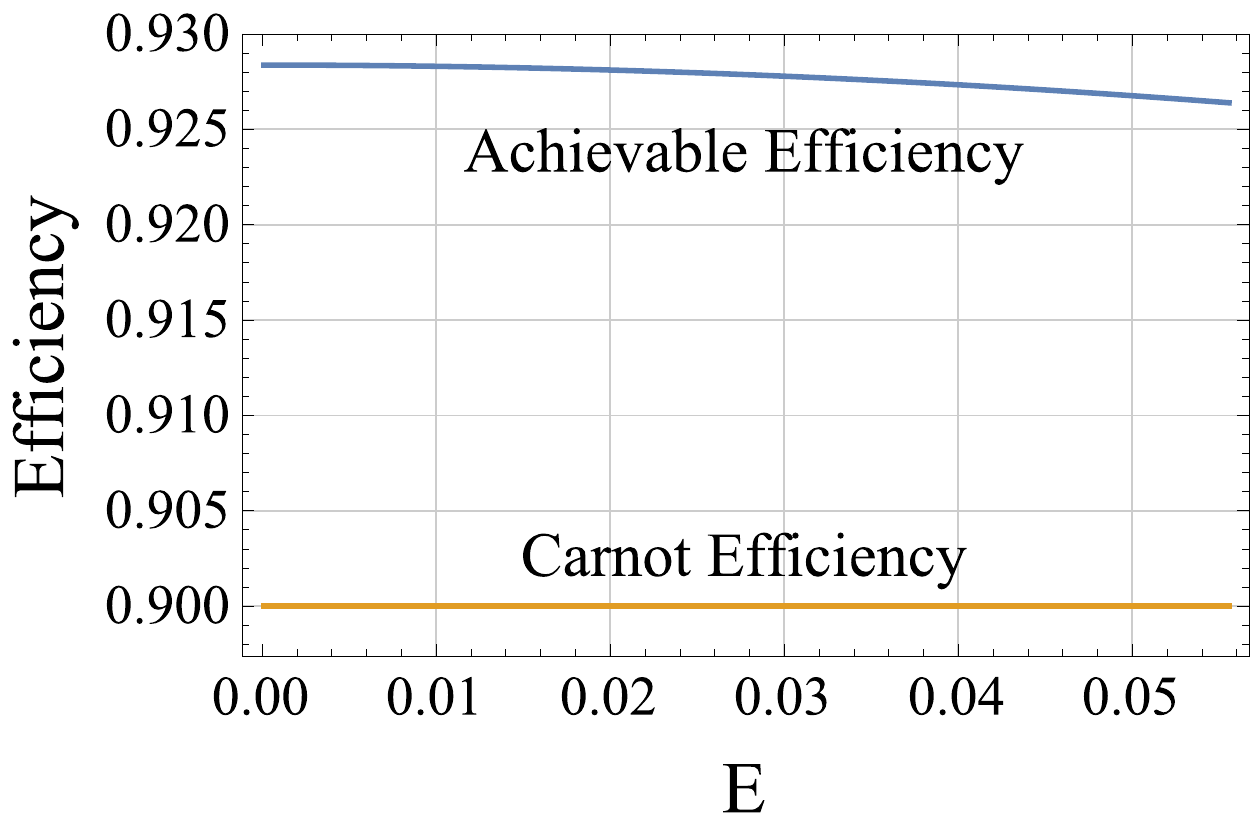}
                \caption{Achievable efficiency versus Carnot efficiency with respect to $E$, with $\alpha^*=1.2$, $\beta_h=1$, $\beta_c=10$, $E=\frac{0.4}{\beta_c-\beta_h}$.}
                \label{fig:bc}
        \end{minipage}
\end{figure}
		
		With this, we can numerically plot out the achievable efficiency as a function of $\beta_c,\beta_h,n,E,\alpha^*$, in the limit where $g\rightarrow 0^+$. It is worth noting that from Eq.~\eqref{eq:mainresult_eff}, we see that the efficiency contains terms that originate from the expression of $\varepsilon_1$ chosen in Eq.~\eqref{eq:mainresult_eps}. It is then, perhaps, unsurprising that we observe the surpassing of Carnot efficiency (for some values of $\alpha^*>1$).
		Indeed, although the average energy change in the battery is positive, i.e. $\Delta W = (1-\varepsilon) \Wext >0$, the change in \emph{free energy} of the battery, 
		\begin{equation}
		\Delta F_\batt = F(\rho_\batt^1) - F(\rho_\batt^0) = \Delta W - \beta_h^{-1} \Delta S
		 \end{equation} 
		 is actually negative. This can be seen when we compute the limit 
		\begin{align*}
		\lim_{g\rightarrow 0^+}	\frac{\Delta F_\batt}{\Delta W} = \lim_{g\rightarrow 0^+}\frac{\Delta W - \beta_h^{-1}\Delta S}{\Delta W} = 1 - \beta_h^{-1} \lim_{g\rightarrow 0^+}\frac{\Delta S}{(1-\varepsilon)\Wext} = -\infty,
		\end{align*}
		where the last limit comes from noting that $\displaystyle\lim_{g\rightarrow 0^+}\varepsilon = 0$, and applying Lemma \ref{lem:free_energy_batt}.

	\subsection{Drawing imperfect work with entropy comparable with $\Wext$}\label{subsec:entropy_comparable}
	In this section we analyze the achievable efficiency when considering the quasi-static limit where 
	\begin{equation}\label{eq:ent_comp}
		\frac{\Delta S}{\Wext} \rightarrow c \qquad \textup{for some}~ c>0.
	\end{equation}	
	 One can see that only certain choices of $\varepsilon (g)$ will lead to having such a limit, which we shall see later in detail on Table \ref{table:2}. We prove that for all choices of $\varepsilon$ such that Eq.~\eqref{eq:ent_comp} is true, one cannot surpass the Carnot efficiency.

\begin{theorem}
Consider a quasi-static heat engine where the failure probability of extracting work is $\varepsilon (g)$, $g$ being the quasi-static parameter (see definition in main text), such that
	\be \label{eq:kappabar_1}
\lim_{g\rightarrow 0^+}\frac{\varepsilon^\kappa(g)}{g}=
\begin{cases}
0 &\mbox{if } \kappa\geq1\\
\infty &\mbox{if } \kappa<1.\\
\end{cases}
\ee
and $\displaystyle\lim_{g\rightarrow 0}\frac{\varepsilon\ln\frac{1}{\varepsilon}}{g} = c >0$. Then the maximum achievable efficiency is upper bounded by the Carnot efficiency.
\end{theorem}
\begin{proof}
Firstly, note that an example for such a choice of $\varepsilon$ can be constructed, i.e. $\varepsilon \ln \frac{1}{\varepsilon} = c\cdot g$. 

We make use of Eq.~\eqref{eq:kappabar_1} to analyze $\Wext$, which is given in Appendix \ref{subsec:technical}. Rewriting Eq.~\eqref{eq:Walp_identicalqubits_quasistatic} by first drawing out a factor of $g$, and neglecting the higher order terms,
		\begin{equation}\label{eq:Wext_as_inf}
			\Wext= g\cdot \inf_{\alpha > 0} \tilde W_\alpha, 
		\end{equation}
where
		 \begin{align}\label{eq:Walp_identicalqubits_quasistatic_eps2}
			\tilde W_\alpha=
			\begin{cases}
			\frac{1}{\beta_h (\alpha-1)} \left[\alpha n B_\alpha - \frac{\varepsilon^\alpha}{g} + \frac{\alpha\varepsilon}{g}\right], &\qquad\mathrm{if } ~\alpha\in(0,\infty)\backslash\lbrace 1 \rbrace,\\[3pt]
			\beta_h^{-1} \left[\displaystyle\lim_{\alpha\rightarrow 1^+} \frac{\alpha n B_\alpha}{\alpha-1} + \frac{\varepsilon\ln\frac{1}{\varepsilon}}{g}\right] = \beta_h^{-1}\left[\displaystyle\lim_{\alpha\rightarrow 1^+} \frac{\alpha n B_\alpha}{\alpha-1} +c\right], &\qquad \alpha=1.
			\end{cases}
		\end{align}
		We are, then, interested in evaluating the minimum over $\tilde W_\alpha$. First of all, note that Eq.~\eqref{eq:kappabar_1} implies that for values of $\alpha\in (0,1)$, the term $\frac{-\varepsilon^\alpha}{g (\alpha-1)}$ goes to infinity as $g\rightarrow 0^+$. This implies that the minimization can be restricted to parameters $\alpha\geq 1$.
		
		On the other hand, for $\alpha >1$ the expression for $\tilde W_\alpha$ can be further simplified in the quasi-static limit,
		\be
			\tilde W_\alpha=
			\frac{\alpha n B_\alpha}{\beta_h (\alpha-1)}   \qquad\mathrm{if } ~\alpha\in(1,\infty).
		\ee
		This is because the terms $\frac{\varepsilon^\alpha}{g}, \frac{\alpha \varepsilon}{g}$ now vanish as $g\rightarrow 0^+$. From this we also see that since $\tilde W_1 > \beta_h^{-1}\displaystyle\lim_{\alpha\rightarrow 1^+} \frac{\alpha n B_\alpha}{\alpha-1}$, and by continuity of the function $\frac{\alpha n B_\alpha}{\alpha-1}$ for $\alpha\in (1,\infty)$, $\tilde W_1$ can also be disregarded in the minimization (see Figure \ref{fig:discontinuity} for a pictorial understanding).
		\begin{figure}[h!]
		\includegraphics[scale=0.5]{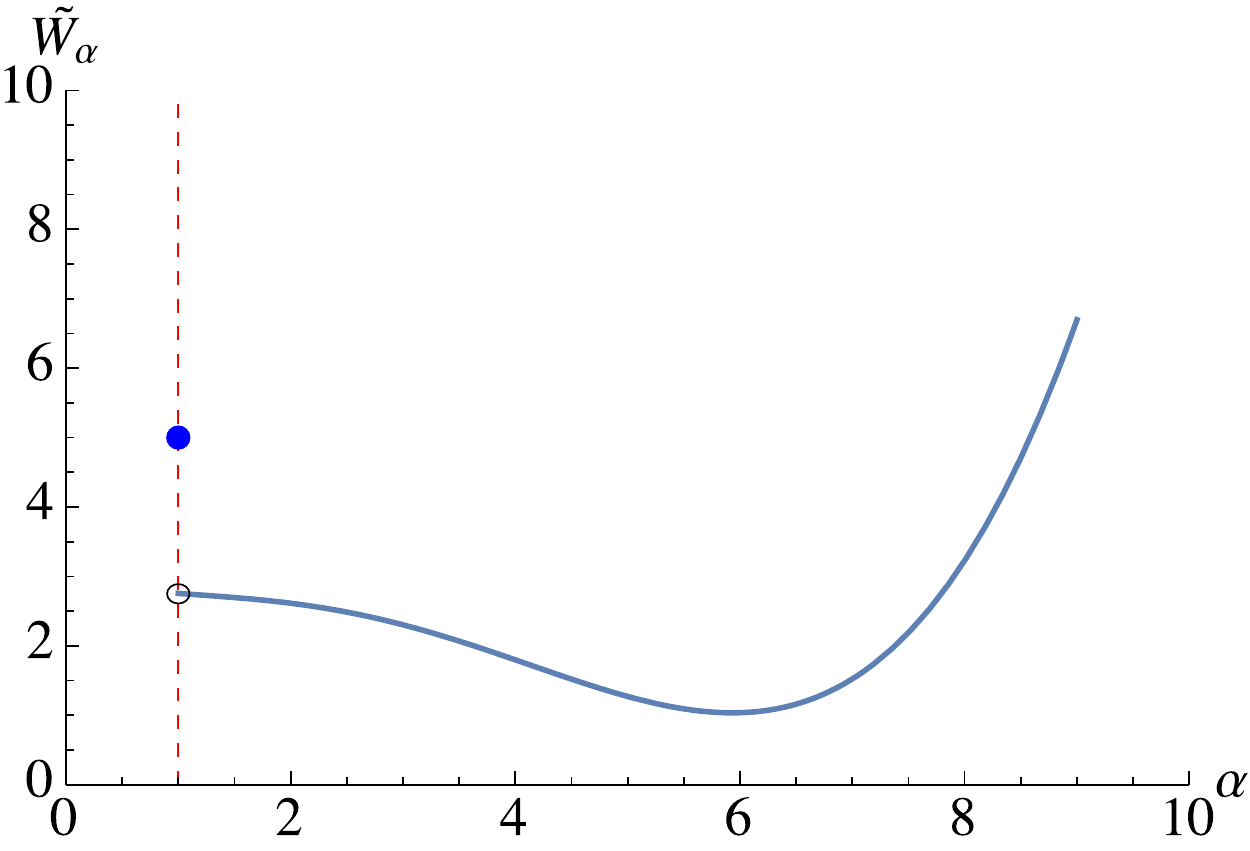}
		\caption{The value of $\tilde W_\alpha$ at $\alpha = 1$ can be ignored while minimizing $\tilde W_\alpha$ over $\alpha\in [1,\infty)$, because the neighbouring values of the function for $\alpha >1$ is lower.}\label{fig:discontinuity}
		\end{figure}
		
		Upon scrutiny, one sees that in the quasi-static limit, the contribution from $\varepsilon$ has dropped out of the expression for $\Wext$. Intuitively this tells us that having such a probability of failure $\varepsilon$ does not help to boost $\Wext$, and in turn the efficiency. In particular, we can upper bound the amount of extractable work by using Eq.~\eqref{eq:Wext_as_inf},
		\begin{align}\label{eq:Wextbound}
		\Wext \leq g\cdot \displaystyle\lim_{\alpha\rightarrow 1^+} \frac{\alpha n B_\alpha}{\alpha-1} = \frac{ng}{\beta_h}	(B_1 + B_1') = \frac{ng}{\beta_h} B_1'.
		\end{align}
		The first equality in Eq.~\eqref{eq:Wextbound} comes by noting that $B_1 = 0$, and therefore applying the L'Hospital rule. The second equality comes again from noting that $B_1 = 0$.
		We can now evaluate an upper bound for the efficiency, 
		\begin{align}\label{eq:etabound}
		\eta^{-1} &= \lim_{g\rightarrow 0^+} 1-\varepsilon + \frac{\Delta C}{W_\textup{ext}}
		 \geq 1 + \frac{ g\cdot \frac{n B_1'}{\beta_c-\beta_h}}{ng \frac{B_1'}{\beta_h}}
		= 1 + \frac{\beta_h}{\beta_c-\beta_h} = \eta_C^{-1}.
		\end{align}
		one finds that the upper bound yields the expression for Carnot efficiency, i.e. $\eta\leq\eta_C$. Eq.~\eqref{eq:etabound} is obtained by applying the identity in Lemma \ref{lem:DeltaC} and the expression in Eq.~\eqref{eq:Wextbound} respectively. This means that for such choices of $\varepsilon(g)$, Carnot efficiency cannot be surpassed.
		\end{proof}

	\section{Technical Lemmas}\label{subsec:technical}
	\subsection{Tools adapted from \cite{WNW15}}\label{subsec:tools}
	In this section, we write out the analytical expressions for the amount of extractable work in the case of a quasi-static heat engine, where the cold bath comprises of $n$ identical systems. In particular, we use the expression of extractable work in Lemma \ref{lem:quasiHE_nidqubits} in order to evaluate the efficiency of our heat engine. The reader is referred to \cite{WNW15} for details of the proof.
	
	Consider a state transition via catalytic thermal operations 
	\begin{equation}\label{eq:st_transition_general}
		\tau_{\beta_\cold}\ot\rho_\batt^0\rightarrow\rho_\cold^1\ot\rho_\batt^1,
	\end{equation}
	where $\tau_{\beta_\cold}$ is the initial state of the cold bath (at inverse temperature $\beta_\cold$), 
	$\rho_\cold^1$ is the final state of the cold bath, and the battery states are given by
	\begin{align}
		\rho_\batt^0 &= \ketbra{E_\batt^j}{E_\batt^j},\label{eq:initial_batt}\\
		\rho_\batt^1 &= \varepsilon \ketbra{E_\batt^j}{E_\batt^j} + (1-\varepsilon) \ketbra{E_\batt^k}{E_\batt^k}.
		\label{eq:final_batt}
	\end{align}
	
	\begin{lemma}\label{lem:quasiHE_nidqubits}
	Consider the state transition described in Eqns.~\eqref{eq:st_transition_general}, \eqref{eq:initial_batt} and \eqref{eq:final_batt}, and assume that the cold bath Hamiltonian is taken to be of $n$ identical systems,
			\begin{equation}\label{eq:tensorprod_Hamiltonian}
				\hat H_\cold=\sum_{i=1}^n \id^{\otimes (i-1)}\otimes \hat H_{\onecold}\otimes\id^{\otimes (n-i)}.
 			\end{equation}
 			Then in the quasi-static limit, where recall that this implies $\rho_\cold^1 = \tau_{\beta_f}$, such that $\beta_f-\beta_c = g\ll 1$, whenever the failure probability $0<\varepsilon \ll1$, the maximum extractable work is 
 			\begin{equation}
 			\Wext =\inf_{\alpha > 0} W_\alpha,
 			\end{equation}
 			where
		 \begin{align}\label{eq:Walp_identicalqubits_quasistatic}		 
			W_\alpha &=
			\begin{cases}
				\frac{1}{\beta_h (\alpha-1)} \left[\alpha n g B_\alpha - \varepsilon^\alpha +\alpha\varepsilon\right] + \bo(g^2)+\bo(\varepsilon^{2\alpha})+\bo(g\varepsilon^\alpha)+\bo(\varepsilon^2) &\mbox{if } \alpha\in(0,\infty)\backslash\lbrace 1 \rbrace,\\
				\displaystyle\lim_{\alpha\rightarrow 1^+}\frac{1}{\beta_h (\alpha-1)} \left[\alpha n g B_\alpha - \varepsilon^\alpha +\alpha\varepsilon\right] + \bo (\varepsilon g)+ \bo (\varepsilon^2 \ln\varepsilon) + \bo (\varepsilon^2)+\bo (g^2) &\mbox{if } \alpha=1,
			\end{cases}
		\end{align}
		and
		\begin{equation}
		B_\alpha =\frac{1}{\displaystyle\sum_i p_i^\alpha q_i^{1-\alpha}} \displaystyle\sum_i p_i^\alpha q_i^{1-\alpha} \left(\langle \hat H_\onecold\rangle_{\beta_c} - E_i \right), \label{eq:Balpha}
		\end{equation}
		where $p_i, q_i$ are the probabilities of thermal states of $\hat H_c$ at inverse temperatures $\beta_c, \beta_h$ respectively.
		In the special case where the cold bath consists of $n$ identical qubits, i.e. $\hat H_{\onecold}=E|1\ra\la 1|$ with $E$ being the energy gap of each qubit, the expression for $B_\alpha$ simplifies to
		\begin{equation}\label{eq:Balp_identicalqubits_quasistatic}
		B_\alpha= \frac{E}{1+e^{\beta_c E}}\cdot \frac{e^{(\beta_h+\alpha\beta_c)E}-e^{(\beta_c+\alpha\beta_h)E}}{e^{\alpha\beta_h E}+e^{(\beta_h+\alpha\beta_c)E}}.
		\end{equation}
	\end{lemma}

		 We also list several expressions that will be useful in deriving our results later.
		 Taking the derivatives of $B_\alpha$ as defined in Eq.~\eqref{eq:Balp_identicalqubits_quasistatic} w.r.t. $\alpha$, we have
		 \begin{align}
		 B_\alpha' =~ \diff{B_\alpha}{\alpha} &= \frac{1}{[e^{\alpha\beta_h E}+e^{(\beta_h+\alpha\beta_c)E}]^2}\cdot  E^2 (\beta_c-\beta_h)\cdot e^{[\beta_h+\alpha(\beta_c+\beta_h)]E}&\label{eq:B_alp_p_identicalqubits_quasistatic}\\
		 &>0 ~~ \textup{whenever}~\beta_c > \beta_h ,\forall\alpha > 0,\vspace{0.1cm}\\
		 B_\alpha'' = \seconddiff{B_\alpha}{\alpha} &=\frac{1}{[e^{\alpha\beta_h E}+e^{(\beta_h+\alpha\beta_c)E}]^3}\cdot  E^3 (\beta_c-\beta_h)^2\cdot e^{[\beta_h+\alpha(\beta_c+\beta_h)]E}\cdot \left[ e^{\alpha\beta_h E}-e^{(\alpha\beta_c+\beta_h)E} \right]& \label{eq:B_alp_pp_identicalqubits_quasistatic}\\
		& <0 ~~ \textup{whenever}~\beta_c > \beta_h ,\forall\alpha > 0.
		 \end{align}

Next, an identity which was proven in \cite{WNW15} will be important for the evaluation of efficiency for a quasi-static heat engine as well. This we present as a lemma here.

\begin{lemma}\label{lem:DeltaC}
Consider a quasi-static heat engine where the cold bath consists of $n$ identical systems (with individual Hamiltonians $\hat H_c$) at inverse temperature $\beta_c$. Denote the inverse temperature of the hot bath as $\beta_h$, and the following function
\begin{align}
\Delta C &:= \tr (\hat H_c \rho_C^1) - \tr (\hat H_c \tau_{\beta_c}).
\end{align}
		Then in the quasi-static limit, where the cold bath final state is a thermal state of inverse temperature $\beta_f = \beta_c - g$, where $0<g\ll 1$,
		\begin{equation}\label{eq:DeltaCexpanded}
		\Delta C = \frac{n B_1'}{\beta_c-\beta_h} \cdot g + \bo(g^2),
		\end{equation}
		where $B_\alpha ' = \diff{B_\alpha}{\alpha}$ and $B_\alpha$ is defined in Eq.~\eqref{eq:Balpha}. 
		\end{lemma}

Lastly, we adopt an observation made in \cite{WNW15} for choices of $\varepsilon (g)$ as a function of the quasi-static parameter $g$. in \cite{WNW15} it is shown that one can characterize any choice of continuous function $\varepsilon (g)$ by the real parameters $\bar\kappa,\sigma \in\mathcal{R}_{\geq 0}$.
\begin{lemma}\label{lemma:existence kappabar}
For every continuous function $\varepsilon(g)>0$ satisfying $\lim_{g\rightarrow 0^+}\varepsilon(g)=0,$ $\exists$ $\bar \kappa\in\mathbb{R}_{\geq 0}$ s.t. 
\be \label{eq:existence lemma}
\delta (\kappa) = \lim_{g\rightarrow 0^+}\frac{\varepsilon^\kappa(g)}{g}=
\begin{cases}
0 &\mbox{if } \kappa>\bar \kappa\\
\sigma\geq 0 &\mbox{if } \kappa=\bar\kappa\\
\infty &\mbox{if } \kappa<\bar\kappa\\
\end{cases}
\ee
where $\bar \kappa=+\infty$ is allowed (that is to say, $\lim_{g\rightarrow 0^+}\frac{\varepsilon^\kappa(g)}{g}$ diverges for every $\kappa$) and $\sigma=+\infty$ is also allowed.
\end{lemma}

Therefore, we summarize results from \cite{WNW15} into the following Table \ref{table:2}, for any continuous function $\varepsilon(g)$ such that $\displaystyle\lim_{g\rightarrow 0}\varepsilon (g)=0$. The first regime, i.e. $\displaystyle\lim_{g\rightarrow 0} \frac{\Delta S}{\Wext} = 0$ is thoroughly investigated in \cite{WNW15}. In this paper, we complete the picture by first analyzing in Section \ref{subsec:mainresult} an example where $\displaystyle\lim_{g\rightarrow 0} \frac{\Delta S}{\Wext} =\infty$, and in Section \ref{subsec:entropy_comparable} investigating the full regime $\displaystyle\lim_{g\rightarrow 0} \frac{\Delta S}{\Wext} = p>0$.
\begin{table}[h!]\label{table:2}
\begin{tabular}{|c|c|c|}
\hline 
\rule{0pt}{4.5ex} 
{} & $\qquad\displaystyle\lim_{g\rightarrow 0} \frac{\Delta S}{\Wext}\qquad$ & $\qquad\qquad\qquad$Characterization$\qquad\qquad\qquad$ \\ [8pt] 
\hline
\multirow{2}{*}{$\quad$Near perfect work$\quad$} & \rule{0pt}{3ex}\multirow{2}{*}{0}& $\bar\kappa\in[0,1)$ \\ [2pt] 
\hhline{~~-} 
\rule{0pt}{4.5ex} 
{} & {} & $\bar\kappa = 1~\wedge~\displaystyle\lim_{g\rightarrow 0} \frac{\varepsilon\ln\frac{1}{\varepsilon}}{g}=0$ \\[10pt] 
\hline 
\rule{0pt}{4.5ex}{} &  $p>0$ & $\bar\kappa=1~\wedge~\displaystyle\lim_{g\rightarrow 0} \frac{\varepsilon\ln\frac{1}{\varepsilon}}{g}= p',~0<p'<\infty$ \\ [10pt] 
\hhline{~--} 
\rule{0pt}{3ex} Imperfect work & \multirow{3}{*}{}  & $\bar\kappa\in (1,\infty)$ \\ [2pt] 
\hhline{~~-} 
{} & $\infty$& \multirow{2}{*}{$\bar\kappa = 1 \wedge \sigma = p'' >0$ }\\ [5pt] 
\hhline{~~~} 
\rule{0pt}{4.5ex}{} &  {} & $\qquad \Big($This implies that $\displaystyle\lim_{g\rightarrow 0} \frac{\varepsilon\ln\frac{1}{\varepsilon}}{g} = \infty \Big) \qquad$ \\ [10pt]
\hline
\end{tabular}
\caption{Each choice of a continuous function $\varepsilon$ such that $\displaystyle\lim_{g\rightarrow 0}\varepsilon = 0$, can lead to different regimes of $\frac{\Delta S}{\Wext}$ in the quasi-static limit, depending on the values of $\bar\kappa, \sigma$ and $\displaystyle\lim_{g\rightarrow 0} \frac{-\varepsilon\ln\varepsilon}{g}$. Recall Lemma \ref{lemma:existence kappabar} for the definitions of $\bar\kappa$ and $\sigma$.}\end{table}

\subsection{Technical Lemmas used for the proof of Theorem \ref{thm:mainresult}}\label{subsub:techlemma}
		Building on the results adapted from \cite{WNW15} and summarized in Section \ref{subsec:tools}, this section contains the technical lemmas and proofs used to develop the proof of Theorem \ref{thm:mainresult}.
		\begin{lemma}\label{lem:inf_not_below_one}
			Given any heat engine, consider the state transition 
			\begin{equation}\label{eq:st_transition}
				\tau_{\beta_\cold}\ot\rho_\batt^0\rightarrow\rho_\cold^1\ot\rho_\batt^1,
			\end{equation}
			where $\rho_\batt^0 = |E_j\ra\la E_j|_\batt ,~\rho_\batt^1=(1-\varepsilon)|E_k\ra\la E_k|_\batt+\varepsilon|E_j\ra\la E_j|_\batt$ respectively, where $\Wext = E_k - E_j$.
			Let $\varepsilon=\varepsilon_1\cdot g$, where note that $\varepsilon_1>0$ is independent of $\alpha$ and $g$. 
			Then there exists $g'>0$ such that for all $0<g\leq g'$, 
			\begin{equation}\label{eq:minimization_of_Walp}
				\Wext= \displaystyle\inf_{\alpha > 0} W_\alpha = \displaystyle\inf_{\alpha> 1} W_\alpha,
			\end{equation}
			where $W_\alpha$ is defined in Eq.~\eqref{eq:Walp_identicalqubits_quasistatic}.
		\end{lemma}
		\begin{proof}
		To prove this, we need only to 1) find $g'$ such that for all $0< g\leq g'$, $\diff{W_\alpha}{\alpha} <0$ whenever $\alpha<1$, and 2) show that the minimum does not occur at $\alpha =1$. Let us do the first. 
		Considering any heat engine with a cold bath that consists of $n$ identical systems, according to Eq.~\eqref{eq:Walp_identicalqubits_quasistatic}, we can evaluate
		\begin{align}\label{eq:derivative_Wpalp_identicalqubits_quasistatic}
			\diff{W_\alpha}{\alpha} 
			&= \frac{1}{\beta_h}\frac{1}{(\alpha-1)^2}
			\left[ (\alpha-1) \left( \alpha n g B_\alpha'+ng B_\alpha-\varepsilon^\alpha\ln\varepsilon+\varepsilon\right)
			- \left(\alpha n gB_\alpha - \varepsilon^\alpha+\alpha\varepsilon\right) \right]\nonumber\\
			&=\frac{1}{\beta_h} \frac{1}{(\alpha-1)^2} 
			\left\lbrace \alpha(\alpha-1) ng B_\alpha' -ngB_\alpha-\varepsilon+\varepsilon^\alpha [1+(1-\alpha)\ln\varepsilon] \right\rbrace\nonumber\\
			&= \frac{g}{\beta_h} \frac{1}{(\alpha-1)^2} \left\lbrace \alpha(\alpha-1)n B_\alpha' -n B_\alpha -\varepsilon_1 +\underbrace{\varepsilon_1^\alpha g^{\alpha-1} [1+(1-\alpha)\ln(\varepsilon_1 g)]}_{f(g,\alpha)} \right\rbrace.
		\end{align}

		Note that for any $0<\alpha<1$, since $\displaystyle\lim_{g\rightarrow 0^+} g^{\alpha-1}=\infty$ and $\displaystyle\lim_{g\rightarrow 0^+} \ln(\varepsilon_1 g)=-\infty$. 
		Therefore, $\displaystyle\lim_{g\rightarrow 0^+}f(g,\alpha) =-\infty$. 
		By the definition of limits, implies that there exists a $g'$ such that for $g\leq g'$, Eq.~\eqref{eq:derivative_Wpalp_identicalqubits_quasistatic} will be negative, implying that the function $W_\alpha$ is monotonically decreasing in the regime $\alpha\in (0,1)$. 
		
 		Secondly, we exclude the point $\alpha =1$ from the minimization, by noting that $W_1 > W_\infty$. Let us first write out the expression for $W_\infty$ as follows:
		\begin{equation}\label{eq:Winftyquasistatic}
			W_\infty = \lim_{\alpha\rightarrow\infty} W_\alpha = \frac{ng}{\beta_h} \left[\lim_{\alpha\rightarrow\infty} \frac{\alpha B_\alpha}{\alpha-1} + \frac{\varepsilon_1}{n}\right] = \frac{ng}{\beta_h} \left[\frac{E}{1+E^{\beta_c E}} + \frac{\varepsilon_1}{n}\right].
		\end{equation}	
		This quantity within the bracket is finite, for finite $E,~\beta_c$. On the other hand, from Eq.~\eqref{eq:Walp_identicalqubits_quasistatic}
		\begin{align}
		W_1 
		&= \frac{1}{\beta_h} \lim_{\alpha\rightarrow 1^+} \frac{\alpha ngB_\alpha-\varepsilon^\alpha +\alpha\varepsilon}{\alpha-1}\label{eq:W1quasistatic}\\
		&= \frac{1}{\beta_h} \lim_{\alpha\rightarrow 1^+} ng B_\alpha' + \alpha ng B_\alpha -\varepsilon^\alpha\ln\varepsilon + \varepsilon\\
		&= \frac{ng}{\beta_h} \left[B_1' - \varepsilon_1 \ln \varepsilon + \varepsilon_1 \right]\\
		&> \frac{ng}{\beta_h} \cdot\varepsilon_1 \cdot\ln \frac{1}{\varepsilon}.\label{eq:W1quasistatic_inequality}
		\end{align}
		The second equality comes by applying L'Hospital rule for differentiation limits, and the third equality comes by substituting $\alpha=1$ into the equation, while noting that $B_1 = 0$, and using $\varepsilon = \varepsilon_1 \cdot g$.
		The last inequality sign comes from noting that $B_1',\varepsilon_1 >0$. For any finite $\varepsilon_1$, we see that in the limit of $g\rightarrow 0^+$, $\varepsilon_1\cdot\ln\frac{1}{\varepsilon}$ tends to infinity, and therefore $W_1/g$ tends to infinity. Comparing Eq.~\eqref{eq:Winftyquasistatic} and \eqref{eq:W1quasistatic_inequality} , we see that in the quasi-static regime, $W_1 > W_\infty$.
		
		Therefore the global minima will not be obtained in the interval $\alpha\in (0,1]$, which in turn implies that 
		\begin{equation}
		\inf_{\alpha>0} W_\alpha = \inf_{\alpha > 1} W_\alpha.
		\end{equation}		
		
		\end{proof}

		With Lemma \ref{lem:inf_not_below_one}, one can dismiss the regime $\alpha\leq 1$ when considering the infimum over $W_\alpha$ in Eq.~\eqref{eq:Walp_identicalqubits_quasistatic}. 
		Note also this implies that while analyzing Eq.~\eqref{eq:Walp_identicalqubits_quasistatic} and \eqref{eq:derivative_Wpalp_identicalqubits_quasistatic},
		the term $\varepsilon^\alpha$ in Eq.~\eqref{eq:Walp_identicalqubits_quasistatic} and $f(g,\alpha)$ in Eq.~\eqref{eq:derivative_Wpalp_identicalqubits_quasistatic} can be omitted as higher order terms, 
		since they vanish more quickly as $g\rightarrow 0^+$ compared to the leading order terms.
		We therefore write out again the form of Eq.~\eqref{eq:derivative_Wpalp_identicalqubits_quasistatic}
		\begin{equation}\label{eq:derivative_Walp_quasistatic_identicalqubits_simple}
			\diff{W_\alpha}{\alpha} = \frac{ng}{\beta_h} \frac{B_\alpha'}{(\alpha-1)^2} \left\lbrace \alpha(\alpha-1) - \frac{B_\alpha}{B_\alpha'} - \frac{\varepsilon_1}{n B_\alpha'} \right\rbrace.
		\end{equation}
		
		From this, we can already understand how $\frac{\Delta S}{\Wext}$ behaves in the quasi-static limit, which we prove in Lemma \ref{lem:free_energy_batt}.
		\begin{lemma}\label{lem:free_energy_batt}
		For any heat engine where $\varepsilon = \varepsilon_1 \cdot g$, with $\varepsilon_1$ independent of $g$, in the quasi-static limit $g\rightarrow 0^+$, we have 
		\begin{equation}
		\lim_{g\rightarrow 0^+}\frac{\Delta S}{\Wext}= \infty.
		\end{equation}
		\end{lemma}
		\begin{proof}
		From Lemma \ref{lem:inf_not_below_one}, and by using Eq.~\eqref{eq:Walp_identicalqubits_quasistatic} we see that for some particular $\alpha_1\in (1,\infty)$,
		\begin{align}
		\Wext &= \frac{1}{\beta_h (\alpha_1 -1)} \left[\alpha n g B_{\alpha_1} - \varepsilon^\alpha +\alpha\varepsilon\right] + \bo(g^2)+\bo(\varepsilon^{2\alpha})+\bo(g\varepsilon^\alpha)+\bo(\varepsilon^2)\\
		&= \frac{g}{\beta_h (\alpha_1 -1)} \left[\alpha_1 n B_{\alpha_1} +\alpha_1\varepsilon_1\right] + \bo(g^{\alpha_1}) + \bo(g^2)+\bo\left(g^{2\alpha_1}\right)+\bo\left(g^{1+\alpha_1}\right)	.	
		\end{align}
		This implies that the leading order term in $\Wext$ is of first order in $g$. On the other hand,
		\begin{align}
		\Delta S &=-\varepsilon\ln\varepsilon - (1-\varepsilon)\ln (1-\varepsilon) \\
		&= -\varepsilon_1\cdot g \ln (\varepsilon_1 \cdot g) - (1-\varepsilon) [-\varepsilon + \bo(\varepsilon^2)] \\
		&= -\varepsilon_1 \cdot g\ln g + \varepsilon_1\ln\varepsilon_1\cdot g + \varepsilon + \bo (\varepsilon^2)+ \bo (\varepsilon^3)\\
		&= -\varepsilon_1 \cdot g\ln g + \bo (g) + \bo (g^2)+ \bo (g^3).
		\end{align}
		The second equality is obtained by substituting $\varepsilon = \varepsilon_1\cdot g$ and writing $\ln (1-\varepsilon) = -\varepsilon +\bo(\varepsilon^2)$ in terms of Taylor expansion. The third equality is obtained by expanding out all the multiplied brackets, while the last equality is obtained by noting that $\bo(\varepsilon) = \bo(g)$, and therefore concluding that the leading order term (which has the slowest convergence rate as $g\rightarrow 0$) is of order $g\ln g$. 
		With this, one can evaluate the limit
		\begin{align}
		\lim_{g\rightarrow 0^+}\frac{\Delta S}{\Wext} &= \lim_{g\rightarrow 0^+} \frac{-\varepsilon_1 \cdot g\ln g + \bo (g) + \bo (g^2)+ \bo (g^3)}{\frac{g}{\beta_h (\alpha_1 -1)} \left[\alpha_1 n B_{\alpha_1} +\alpha_1\varepsilon_1\right] + \bo(g^{\alpha_1}) + \bo(g^2)+\bo\left(g^{2\alpha_1}\right)+\bo\left(g^{1+\alpha_1}\right)}\\
		& = \lim_{g\rightarrow 0^+} \frac{-\varepsilon_1 \cdot \ln g + \bo (1) + \bo (g)+ \bo (g^2)}{\frac{1}{\beta_h (\alpha_1 -1)} \left[\alpha_1 n B_{\alpha_1} +\alpha_1\varepsilon_1\right] + \bo(g^{\alpha_1-1}) + \bo(g)+\bo\left(g^{2\alpha_1-1}\right)+\bo\left(g^{\alpha_1}\right)	}\\
		&=\infty.
		\end{align}
		The second equality is obtained by dividing both numerator and denominator with $g$. Then we see that in the numerator, $-\varepsilon_1 \cdot\ln g$ goes to infinity, while the other terms remain finite. On the other hand, the denominator goes to a finite constant. Therefore, we conclude that $\lim_{g\rightarrow 0^+}\frac{\Delta S}{\Wext} = \infty$. 
		\end{proof}
		
		\textbf{From here onwards, we focus our analysis to the case where the cold bath consists of qubits. Therefore, $B_\alpha$ is given by Eq.~\eqref{eq:Balp_identicalqubits_quasistatic}, and $B_\alpha',~B_\alpha''$ in Eq.~\eqref{eq:B_alp_p_identicalqubits_quasistatic}, \eqref{eq:B_alp_pp_identicalqubits_quasistatic} respectively.}
		
		The next Lemmas \ref{lem:falpha_is_concave} and \ref{lem:notmorethantwostpt} would establish a useful property of $\diff{W_\alpha}{\alpha}$, namely that this function has not more than 3 roots in the regime $\alpha\in (1,\infty)$, i.e. $W_\alpha$ does not have more than 3 stationary points. 
		Then in Lemma \ref{lem:Winfty_approachedfromabove} we show how the value of $\lim_{\alpha\rightarrow\infty} W_\alpha$ is approached. 
		
		\begin{lemma}\label{lem:falpha_is_concave}
		Consider the function $f(\alpha):=\alpha(\alpha-1)-\frac{B_\alpha}{B_\alpha'} - \frac{\varepsilon_1}{n B_\alpha'}$, which is found in the R.H.S. of Eq.~\eqref{eq:derivative_Walp_quasistatic_identicalqubits_simple}. Then its first derivative w.r.t. $\alpha$, $f'(\alpha) = \diff{f(\alpha)}{\alpha}$ is strictly concave in the domain $\alpha\in (1,\infty)$. This also implies that $f(\alpha)$ has at most 3 roots in the regime $\alpha\in (1,\infty)$.
		\end{lemma}
		\begin{proof}
		Note that $f'(\alpha) = g'(\alpha)+\frac{\varepsilon_1}{n}\frac{B_\alpha''}{B_\alpha'^2}$, where $g'(\alpha) = \diff{}{\alpha} [\alpha (\alpha-1)-\frac{B_\alpha}{B_\alpha'}]$. It has been shown in Lemma 12, Supplementary Material of \cite{WNW15} that $g'(\alpha)$ is a strictly concave function. On the other hand, by using the definitions in Eq.~\eqref{eq:B_alp_p_identicalqubits_quasistatic} and \eqref{eq:B_alp_pp_identicalqubits_quasistatic},  one can evaluate the second derivative of 
		\begin{align*}
		\seconddiff{}{\alpha} \frac{B_\alpha''}{B_\alpha'^2} = (\beta_c-\beta_h)^2 e^{-[\beta_h +\alpha(\beta_c+\beta_h)]E}E\cdot \left[e^{2\alpha\beta_h E}-e^{2(\alpha\beta_c+\beta_h)E}\right].
		\end{align*}
		All the terms in the equation above are positive, except for the last term which is always negative when $\beta_h<\beta_c$. Therefore, the function $\frac{B_\alpha''}{B_\alpha'^2}$ is strictly concave as well. This implies that $f'(\alpha)$ is the addition of two strictly concave functions, and therefore is also strictly concave itself. 
		\end{proof}
		
		One can apply Lemma \ref{lem:falpha_is_concave} to analyze the function $W_\alpha$ to show that it does not have more than 3 stationary points. 
		\begin{lemma}[$W_\alpha$ has not more than 3 stationary points]\label{lem:notmorethantwostpt} 
		Consider the continuous function $W_\alpha = \frac{ng}{\beta_h (\alpha-1)} \left[\alpha B_\alpha +\frac{\alpha\varepsilon_1}{n}\right]$ in the regime $\alpha\in (1,\infty)$. Then the equation $\diff{W_\alpha}{\alpha}=0$ has at most 3 roots, i.e. the function $W_\alpha$ has not more than 3 stationary points. 
		\end{lemma}
		\begin{proof}
		Let us begin by writing out the function 
		\begin{align*}
		\diff{W_\alpha}{\alpha} = \frac{ng}{\beta_h} \frac{1}{(\alpha-1)^2}B_\alpha' \left\lbrace \alpha(\alpha-1) -n \frac{B_\alpha}{B_\alpha'} - \frac{\varepsilon_1}{n B_\alpha'} \right\rbrace.
		\end{align*}
		Since from the expression in Eq.~\eqref{eq:B_alp_p_identicalqubits_quasistatic}, we see that $B_\alpha' >0$ whenever $\beta_h<\beta_c$, by Lemma \ref{lem:falpha_is_concave}, we know that $\diff{W_\alpha}{\alpha} $ can have at most 3 roots.
		\end{proof}

		\begin{lemma}[$W_\infty$ is approached from above]\label{lem:Winfty_approachedfromabove}
		Consider the continuous function $W_\alpha = \frac{ng}{\beta_h (\alpha-1)} \left[\alpha B_\alpha +\frac{\alpha\varepsilon_1}{n}\right]$ in the regime $\alpha\in (1,\infty)$. Then the limit $\displaystyle\lim_{\alpha\rightarrow\infty} \Walpha$ exists and is approached from above.
		\end{lemma}
		\begin{proof}
		We have seen from Eq.~\eqref{eq:Winftyquasistatic} that $\lim_{\alpha\rightarrow\infty} W_\alpha$ exists and is some finite number. We then only need to prove that in the limit of large $\alpha$, the quantity $\diff{W_\alpha}{\alpha} <0$. This can be seen from Eq.~\eqref{eq:derivative_Walp_quasistatic_identicalqubits_simple}, which we rewrite here again
		\begin{equation}\label{eq:lemWinfty}
		\diff{W_\alpha}{\alpha} = \frac{ng}{\beta_h} \frac{1}{(\alpha-1)^2}\left\lbrace \alpha(\alpha-1) B_\alpha'- B_\alpha - \frac{\varepsilon_1}{n} \right\rbrace.
		\end{equation}
		
		Let us compare the terms in the large bracket of the R.H.S. The first term 
		\begin{equation}
		\alpha (\alpha-1) B_\alpha' = \alpha (\alpha-1) E^2 (\beta_c-\beta_h) e^{-\beta_h E} e^{-\alpha (\beta_c-\beta_h)E}
		\end{equation}
		has a quadratic term in $\alpha$ multiplied by a term which decreases exponentially in $\alpha$, i.e. $\lim_{\alpha\rightarrow\infty} \alpha (\alpha-1) B_\alpha' = 0$. On the other hand, the remaining terms
		\begin{equation}
		\lim_{\alpha\rightarrow\infty} -B_\alpha - \frac{\varepsilon_1}{n} = - \left[ \frac{E}{1+e^{\beta_c E}} + \frac{\varepsilon_1}{n} \right] <0.
		\end{equation}
		Since for large $\alpha \gg 1$, the multiplicative factor in Eq.~\eqref{eq:lemWinfty} is positive, we have that $\diff{W_\alpha}{\alpha} <0$. This implies that the function $W_\alpha$ approaches the limit $\alpha\rightarrow\infty$ from above.
		\end{proof}
		
 		\begin{lemma}\label{lem:onlyonelocalmin}
 		The function $W_\alpha$ does not have more than one distinct local minimas in the regime $\alpha\in (1,\infty)$.
 		\end{lemma}	 	
 		\begin{proof}
			By Lemma \ref{lem:notmorethantwostpt}, we know that the function $W_\alpha$ has at most 3 stationary points in the regime $\alpha\in (1,\infty)$. 
			Firstly, suppose that $W_\alpha$ has only 1 or 2 stationary points. Then it is clear that there cannot exist two distinct local minimas, since for a continuous function with two local minimas, there has to be at least another local maxima in between, which is also a stationary point. 
			
			Now, suppose that $W_\alpha$ has 3 stationary points, found at $1 <\alpha_1 < \alpha_2 <\alpha_3 < \infty$ respectively. Note that two neighbouring stationary points cannot both correspond to local minimas, as reasoned out in the previous paragraph. Therefore, the only way for there to exist 2 local minimum points, is to have $\alpha_1, \alpha_3$ corresponding to local minimas. If there are no more stationary points in the regime $\alpha > \alpha_3$, then the function $W_\alpha$ can only be non-decreasing, and the limit $\alpha\rightarrow\infty$ has to be approached from below. However, by Lemma \ref{lem:Winfty_approachedfromabove} we know that this cannot be true.
				
			This establishes the fact that $W_\alpha$ does not have two distinct local minimas. Therefore, it implies that whenever we find some $\alpha^*$ corresponding to a local minima, it will be the unique local minima of the entire function. This simplifies the minimization of $W_\alpha$ in Eq.~\eqref{eq:minimization_of_Walp} to comparing $W_{\alpha^*}$ with $W_\infty$.
		
		\end{proof} 			

		In the next lemma, we then prove that by making use of our liberty to choose $\varepsilon_1$, we can design it such that $\Wext=\inf_{\alpha> 1}W_\alpha$ is obtained at any $\alpha^*$ we desire.
		
		\begin{lemma}(Conditions for positive $\varepsilon_1$)\label{lem:condition_positive_eps1}
			Consider the function 
			\begin{equation}\label{eq:def_eps1}
			 	\varepsilon_1 (a,n) := n [a(a-1)B_{a}' - B_{a}].
			 \end{equation} When the condition 
			\begin{equation}\label{eq:condition_positive_eps1}
				E<\frac{2}{\beta_c-\beta_h}\frac{1+e^{\beta_c E}}{e^{\beta_c E}-1}
			\end{equation}
			holds, then there exists some $\alpha^*>1$ such that $\varepsilon_1 (\alpha^*,n) >0$.
		\end{lemma}
		\begin{proof}
			We begin by noting that $\varepsilon(1,n)=0$ for any $n$. A Taylor's expansion around $a=1$ would	determine the positivity of $\varepsilon(a,n)$ for $a=1+\delta$ where $\delta\ll 1$. Therefore, we calculate the derivative $\diff{\varepsilon_1(a,n)}{a}$
			\begin{equation}\label{eq:derivative_eps1}
				\diff{\varepsilon_1(a,n)}{a} = n [(a-1)B_a'+aB_a'+a(a-1)B_a''-B_a'] = n(a-1)[2B_a'+aB_a''].
			\end{equation}
			It is easy to see from Eq.~\eqref{eq:derivative_eps1} that $\diff{\varepsilon_1(a,n)}{a}\Big|_{a=1} =0$.
			Therefore, the term that determines positivity of $\varepsilon_1 (a,n)$ around $a=1$ is the second derivative,
			\begin{equation}
				\seconddiff{\varepsilon_1(a,n)}{a} = n[2B_a'+aB_a''+(a-1)(2B_a''+B_a''+aB_a''')].
			\end{equation}
			The quantity $\seconddiff{\varepsilon_1(a,n)}{a}\Big|_{a=1} = n [2B_1'+B_1'']$ we can expressed in a simplified form,
			\begin{equation}
				\seconddiff{\varepsilon_1(a,n)}{a}\Big|_{a=1} = \frac{n (\beta_c-\beta_h) e^{(\beta_c + 3\beta_h)E} E^2}{[e^{\beta_h E}+ e^{(\beta_c+\beta_h)E}]^3} \left[ 2 +(\beta_c-\beta_h)E+e^{\beta_c E} (2 - \beta_c E +\beta_h E ) \right].
			\end{equation}
			For this to be positive, it implies that $ 2 +(\beta_c-\beta_h)E+e^{\beta_c E} (2 - \beta_c E +\beta_h E )>0$.
			Rearranging terms, we find
		\begin{align}
			(\beta_c-\beta_h)E (1-e^{\beta_c E}) &> -2 (1+e^{\beta_c E})\\
			E &< \frac{2}{\beta_c-\beta_h} \frac{e^{\beta_c E}+1}{e^{\beta_c E}-1}.
		\end{align}
		\end{proof}
		
		\begin{lemma}\label{lem:alpstisstationarypoint}
		Consider the function $W_\alpha$ as described in Eq.~\eqref{eq:Walp_identicalqubits_quasistatic}. When $\varepsilon_1$ is given by Eq.~\eqref{eq:def_eps1} for some $a=\alpha^*>1$, then
		$\diff{W_\alpha}{\alpha}\Big|_{\alpha=\alpha^*}=0$.
		\end{lemma}
		\begin{proof}
		To see this, let us write out the final form of the first derivative of $W_\alpha$ in Eq.~\eqref{eq:derivative_Wpalp_identicalqubits_quasistatic},
		\begin{equation}
		\diff{W_\alpha}{\alpha}= \frac{g}{\beta_h} \frac{1}{(\alpha-1)^2} \left\lbrace \alpha(\alpha-1)n B_\alpha' -n B_\alpha -\varepsilon_1 +f(g,\alpha) \right\rbrace,
		\end{equation}
		where $f(g,\alpha)\rightarrow 0$ as $g\rightarrow 0$ for $\alpha >1$, and now onwards can be neglected from our derivations since we consider the  quasi-static limit. 
		Substituting Eq.~\eqref{eq:def_eps1} into the equation above gives us 0 when $\alpha=\alpha^*$. This concludes the proof.		
		\end{proof}
		
So far, for a specific design of $\varepsilon_1$, we've found conditions expressed in Eq.~\eqref{eq:condition_positive_eps1} such that $\varepsilon_1>0$ and $W_{\alpha^*}$ is a stationary point. Next, we can write down further conditions for when given $\alpha^*$ and $\varepsilon_1(\alpha^*,n)$ as defined in Lemma \ref{lem:condition_positive_eps1}, one can now find conditions on $E$ such that $W_\alpha$ not only is a stationary point, but also a local minima.

		\begin{lemma}\label{lem:seconddiff>0}
		Consider the functions 
		\begin{equation}\label{eq:Walp_d1}
			\diff{W_\alpha}{\alpha} = \frac{g}{\beta_h} \frac{1}{(\alpha-1)^2} 
			\left[ \alpha(\alpha-1)n B_\alpha' -n B_\alpha -\varepsilon_1 \right],
		\end{equation}
		and 
		\begin{equation}
			B_\alpha= \frac{E}{1+e^{\beta_c E}}\cdot 
			\frac{e^{(\beta_h+\alpha\beta_c)E}-e^{(\beta_c+\alpha\beta_h)E}}{e^{\alpha\beta_h E}+e^{(\beta_h+\alpha\beta_c)E}},
		\end{equation}
		If the following condition holds:
		\begin{equation}\label{eq:lemma8}
		E<\frac{1}{\beta_c-\beta_h},
		\end{equation}
		there one can find some $\alpha^* >1$ in the vicinity of $\alpha=1$ such that when we define $\varepsilon_1 (\alpha^*,n) := n [\alpha^*(\alpha^*-1)B_{\alpha^*}' - B_{\alpha^*}]$, then $\varepsilon_1 (\alpha^*,n) >0$. Furthermore if $1<\alpha^*<2$ is chosen, then
		\begin{equation}
		\seconddiff{W_\alpha}{\alpha}\Bigg|_{\alpha=\alpha^*}>0.
		\end{equation}
		\end{lemma}
		
		\begin{proof}
		We first note that if $E<\frac{1}{\beta_c-\beta_h}$, then Eq.~\eqref{eq:condition_positive_eps1} holds and therefore by Lemma \ref{lem:condition_positive_eps1}, one can choose some $\alpha^*>1$ and close to 1 such that $\varepsilon_1(\alpha^*,n)>0$. Next, we calculate the analytical expression of $\seconddiff{W_\alpha}{\alpha}$ in terms of first and second derivatives of $B_\alpha$. Differentiating Eq.~\eqref{eq:Walp_d1},
		\begin{align}
		\seconddiff{W_\alpha}{\alpha} 
		&= \frac{g}{\beta_h}\frac{1}{(\alpha-1)^4} 
		\lbrace
		(\alpha-1)^2[(\alpha-1)nB_\alpha'+\alpha n B_\alpha'+\alpha(\alpha-1)nB_\alpha''-nB_\alpha']
		-2(\alpha-1)[\alpha(\alpha-1)nB_\alpha'-nB_\alpha-\varepsilon_1]
		\rbrace\nonumber\\
		&=\frac{g}{\beta_h}\frac{1}{(\alpha-1)^3} 
		\lbrace
		(\alpha-1)[2(\alpha-1)nB_\alpha'+\alpha(\alpha-1)nB_\alpha'']-2[\alpha(\alpha-1)nB_\alpha'-nB_\alpha-\varepsilon_1]
		\rbrace\nonumber\\
		&=\frac{g}{\beta_h}\frac{1}{(\alpha-1)^3} 
		\lbrace
		n(\alpha-1)^2[2B_\alpha'+\alpha B_\alpha'']-2[\alpha(\alpha-1)nB_\alpha'-nB_\alpha-\varepsilon_1]
		\rbrace.\label{eq:approx_Walpha_seonddiff}
		\end{align}
		Substituting $\alpha=\alpha^*$ into Eq.~\eqref{eq:approx_Walpha_seonddiff}, one sees that the last term vanishes, and therefore
		\begin{equation}
		\seconddiff{W_\alpha}{\alpha}\Bigg|_{\alpha=\alpha^*}=
		\frac{ng}{\beta_h}\frac{1}{(\alpha^*-1)} 
		[2B_{\alpha^*}'+\alpha B_{\alpha^*}''].
		\end{equation}
		Since $\alpha^*>1$, we see that to guarantee positivity of Eq.~\eqref{eq:approx_Walpha_seonddiff} is equivalent to showing that the last term $2B_{\alpha^*}'+\alpha B_{\alpha^*}''$ is strictly positive.
		To do so, we evaluate the terms $B_\alpha'$ and $B_\alpha''$. By both hand derivation and Mathematica, we obtain the expressions		
		\begin{equation}
			B_\alpha' =\frac{1}{[e^{\alpha\beta_h E}+e^{(\beta_h+\alpha\beta_c)E}]^2}\cdot  E^2 (\beta_c-\beta_h)\cdot e^{[\beta_h+\alpha(\beta_c+\beta_h)]E}
		\end{equation}		
		and 
		\begin{equation}
		B_\alpha'' =\frac{1}{[e^{\alpha\beta_h E}+e^{(\beta_h+\alpha\beta_c)E}]^3}\cdot  E^3 (\beta_c-\beta_h)^2\cdot e^{[\beta_h+\alpha(\beta_c+\beta_h)]E}\cdot \left[ e^{\alpha\beta_h E}-e^{(\alpha\beta_c+\beta_h)E} \right].
		\end{equation}
		
		One can then calculate the last term in Eq.~\eqref{eq:approx_Walpha_seonddiff}, which we again obtain a simplified expression via Mathematica,
		\begin{equation}
		2B_{\alpha^*}'+\alpha B_{\alpha^*}''= \underbrace{\frac{(\beta_c-\beta_h)E^2}{[e^{\alpha\beta_h E}+e^{(\beta_h+\alpha\beta_c)E}]^3}}_{>0}\cdot~\underbrace{e^{[\beta_h+\alpha(\beta_c+\beta_h)]E}}_{>0}\cdot f(\alpha^*),
		\end{equation}
		where 
		\begin{align}
		f(\alpha^*) 
		&:= e^{\alpha^*\beta_h E} [2+\alpha^*(\beta_c-\beta_h)E]
		+e^{(\alpha^*\beta_c+\beta_h) E}[2-\alpha^*(\beta_c-\beta_h)E]\\
		&~= \underbrace{2 \left[e^{\alpha^*\beta_h E}+e^{(\alpha^*\beta_c+\beta_h) E}\right]}_{>0} + \underbrace{\alpha^* (\beta_c-\beta_h) E \left[e^{\alpha^*\beta_h E}-e^{(\alpha^*\beta_c+\beta_h) E}\right] }_{<0}.
		\end{align}
		Note that the second term is always negative because $\alpha^*>1$ and $\beta_c>\beta_h$. Therefore, to lower bound $f(\alpha^*)$ we want to upper bound the factor $\alpha^* (\beta_c-\beta_h) E$. By letting $1<\alpha^*<2$ and $E<\frac{1}{\beta_c-\beta_h}$, one can have $\alpha^* (\beta_c-\beta_h) E<2$, which gives 
		\begin{equation}
		2B_{\alpha^*}'+\alpha B_{\alpha^*}'' > 2\left[e^{\alpha^*\beta_h E}+e^{(\alpha^*\beta_c+\beta_h) E}\right] + 2 \left[e^{\alpha^*\beta_h E}-e^{(\alpha^*\beta_c+\beta_h) E}\right] = 4e^{\alpha^*\beta_h E} >0.
		\end{equation}
		Note that the constraints on $\alpha^*$ and $E$ are not necessary, however sufficient and takes a relatively simple form.
		\end{proof}

		Finally, for $W_{\alpha^*}$ to be the global minima, since we have already seen that $W_1>W_\infty$ and cannot be the global minima, therefore we need one last condition: that $W_{\alpha^*}<W_\infty$. In the next lemma, we again develop conditions such that this is true.
		
		\begin{lemma}\label{lem:Walpst<Winfty}
		Suppose $\alpha^* E < \frac{1}{\beta_c-\beta_h}$. Then for $\varepsilon_1 (\alpha^*,n)$ defined as in Eq.~\eqref{eq:def_eps1}, we have that $W_{\alpha^*} < W_\infty$.
		\end{lemma}
		\begin{proof}
		To do so, we write out the expressions for $W_{\alpha^*}$ and $W_\infty$. The former can be written using Eq.~\eqref{eq:Walp_identicalqubits_quasistatic}, while the later has been derived in Eq.~\eqref{eq:Winftyquasistatic}:
		\begin{align}
			W_{\alpha} & = \frac{ng}{\beta_h} \frac{\alpha^*}{\alpha^*-1} [B_{\alpha^*} + \alpha^* (\alpha^*-1) B_{\alpha^*}' - B_{\alpha^*}] = \frac{ng}{\beta_h} \alpha^{*2} B_{\alpha^*}'\\
			W_\infty & = \frac{ng}{\beta_h} \left[\frac{E}{1+E^{\beta_c E}} + \frac{\varepsilon_1}{n}\right].
		\end{align}
		For $W_{\alpha^*}<W_\infty$, this means
		\begin{align}
			\alpha^{*2} B_{\alpha^*}' < \frac{E}{1+E^{\beta_c E}} &+ \alpha^* (\alpha^*-1) B_{\alpha^*}' - B_{\alpha^*}\\
			\frac{E}{1+E^{\beta_c E}} & - \alpha^* B_{\alpha^*}' - B_{\alpha^*} > 0.\label{eq:WalpstleqWinfty}
		\end{align}
		Expanding Eq.~\eqref{eq:WalpstleqWinfty}, and using the shorthand $A := e^{\alpha^*\beta_h E}+e^{(\beta_h+\alpha^*\beta_c)E}$ we obtain
		\begin{align}
			&\frac{E}{1+E^{\beta_c E}} - \frac{(\beta_c-\beta_h)\alpha^{*2} E^2}{A^2} e^{[\beta_h + \alpha^* (\beta_c+\beta_h)]E} - \frac{E}{1+E^{\beta_c E}} \frac{e^{(\beta_h+\alpha^*\beta_c) E}+e^{(\beta_c+\alpha^*\beta_h)E}}{A}\\
			&= \frac{E}{1+E^{\beta_c E}} \frac{1}{A^2} 
			\left\lbrace A^2 - \alpha^*E (\beta_c+\beta_h) (1+e^{\beta_c E}) e^{[\beta_h + \alpha^* (\beta_c+\beta_h)]E} 
			- A\left[e^{\alpha^*\beta_h E}+e^{(\beta_h+\alpha^*\beta_c)E}\right] \right\rbrace\\
			&= \frac{E}{1+E^{\beta_c E}} \frac{1}{A^2} 
			\left\lbrace A \left[e^{\alpha^*\beta_h E}+E^{(\alpha^*\beta_h+\beta_c)E}\right] 
			- \alpha^*E (\beta_c+\beta_h) \left(1+e^{\beta_c E}\right) e^{[\beta_h + \alpha^*E (\beta_c+\beta_h)]E}  \right\rbrace\label{eq:derivation_mid}\\
			&=  \frac{E}{A^2} \cdot e^{\alpha^*\beta_h E} \left\lbrace A - \alpha^*E (\beta_c+\beta_h)  e^{(\beta_h + \alpha^* \beta_c)E} \right\rbrace\\
			& = \frac{E}{A^2} \cdot e^{\alpha^*\beta_h E} \left\lbrace e^{\alpha^*\beta_h E} + e^{(\beta_h+\alpha^*\beta_c)E} \left[ 1-\alpha^* E(\beta_c-\beta_h) \right] \right\rbrace.
		\end{align}
		The calculation above can be checked as follows: the first equality is obtained by taking out a common factor from all the three terms. The second equality focuses on the large bracket, and combines the first and third terms by expanding one of the $A$ in the first term. In the third equality, one recognizes more common factors in Eq.~\eqref{eq:derivation_mid}, and therefore pulls out $e^{\alpha^*\beta_h E} \cdot (1+e^{\beta_c E})$. The fourth equality is obtained by expanding $A$, while regrouping some terms. To demand that $W_{\alpha^*} < W_\infty$, implies that we demand 
		\begin{equation}\label{eq:forWalpstleqWinfty}
		e^{\alpha^*\beta_h E} + e^{(\beta_h+\alpha^*\beta_c)E} \left[ 1-\alpha^* E(\beta_c-\beta_h) \right] >0.
		\end{equation}
		Rearranging Eq.~\eqref{eq:forWalpstleqWinfty}, we have
		\begin{align}
			e^{\alpha^*\beta_h E} &> e^{(\beta_h+\alpha^*\beta_c)E} \left[ \alpha^* E(\beta_c-\beta_h)-1 \right].
		\end{align}
		One can continue to simplify the expression by bringing $e^{(\beta_h+\alpha^*\beta_c)E}$, and subsequently the $-1$ to the L.H.S., yielding
		\begin{equation}
		1+e^{-\alpha^*(\beta_c-\beta_h) E} e^{-\beta_h E} >  \alpha^* E(\beta_c-\beta_h).
		\end{equation}
		Then finally, one obtains an expression for $\alpst E$:
		\begin{equation}
		\alpha^* E < \frac{1+e^{-\alpha^*(\beta_c-\beta_h) E} e^{-\beta_h E} }{\beta_c-\beta_h}.
		\end{equation}	
		
		Since $\beta_c-\beta_h >0$, and we have that $e^{-\alpha^*(\beta_c-\beta_h) E} e^{-\beta_h E}>0$,  therefore as long as $\alpha^* E < \frac{1}{\beta_c-\beta_h}$, Eq.~\eqref{eq:WalpstleqWinfty} is satisfied and $W_{\alpha^*}<W_\infty$. This concludes our proof.
		\end{proof}
		
		Lemma \ref{lem:inf_not_below_one}, \ref{lem:onlyonelocalmin}, \ref{lem:condition_positive_eps1}, \ref{lem:alpstisstationarypoint} and \ref{lem:seconddiff>0} together presents a set of mathematical conditions such that $\alpha^*$ can be chosen such that $W_\alpha$ has a \emph{global} minima at $W_{\alpha^*}$. This is presented in Eq.~\eqref{eq:finalcond} of Corollary \ref{cor:HamiltonianConditionsForAlphaStarMinimization}.

\end{document}